\documentclass[12pt]{amsart}
\usepackage{amssymb}
\usepackage{amsfonts}
\usepackage{amssymb,latexsym}
\usepackage{enumerate}
\usepackage{mathrsfs}
\usepackage{url}
\usepackage{rotating}
\usepackage{lscape}
\usepackage{hyperref}
\allowdisplaybreaks

\makeatletter
\@namedef{subjclassname@2020}{%
	\textup{2020}Mathematics Subject Classification}
\makeatother

\hypersetup{
  colorlinks=true,
  linkcolor=blue,
  citecolor=blue,
  urlcolor=blue
}

\newcommand{\Z}{\mathbb{Z}}
\newcommand{\Q}{\mathbb{Q}}

\DeclareMathOperator{\supp}{supp}
\DeclareMathOperator{\dist}{dist}

\newcommand{\Zp}{\mathbb{Z}_p}
\newcommand{\Qp}{\mathbb{Q}_p}
\newcommand{\Cp}{\mathbb{C}_p}

\newcommand{\cL}{\mathbf{L}}

\newtheorem{theorem}{Theorem}[section]
\newtheorem{lemma}[theorem]{Lemma}
\newtheorem{proposition}[theorem]{Proposition}

\theoremstyle{definition}
\newtheorem{remark}[theorem]{Remark}
\newtheorem{definition}[theorem]{Definition}

\usepackage{array}
\def\bar{\tilde}
\def\ga{\gamma}

\def\d{\mathrm}

\numberwithin{equation}{section} 

\begin{document}

\def\t{\widetilde}
\def\tilde{\widetilde}
\def\ga{\gamma}
\def\d{{\rm d}}

\title[$p$-adic spectral zeta functions]
{$p$-adic spectral zeta functions via the inverse Stieltjes transform}

\author{Su Hu}
\address{Department of Mathematics, South China University of Technology, Guangzhou, Guangdong 510640, China}
\email{mahusu@scut.edu.cn}

\author{Min-Soo Kim}
\address{Department of Mathematics Education, Kyungnam University, Changwon, Gyeongnam 51767, Republic of Korea}
\email{mskim@kyungnam.ac.kr}

\dedicatory{Dedicated to the memory of Lev Genrikhovich Shnirelman (1905--1938)}
\begin{abstract}
Spectral zeta functions provide a standard tool for regularizing determinants of differential operators in quantum physics. In a previous paper (J. Math. Phys. 66: 083505, 2025), we introduced a $p$-adic spectral zeta function for discrete spectra via a locally analytic interpolation function. In this paper we extend the framework to continuous spectra using the inverse Stieltjes transform. Starting from the resolvent of a bounded operator, we construct a generalized distribution, called the $p$-adic spectral distribution, and define bosonic and fermionic zeta functions and functional determinants. We establish their analytic properties, special value formulas, and Stirling expansions, and show that our previous framework is embedded into the present one in the discrete case. In this approach, the bosonic and fermionic cases exhibit an interesting symmetry. As an application, we consider the position operator in $p$-adic quantum mechanics on an arbitrary compact subset of $\mathbb{C}_{p}$.
\end{abstract}
\maketitle

\section{Introduction}\label{sec:intro}

Throughout this paper, we use the following notation:
\[
\begin{aligned}
\mathbb{N}\;&:\ \text{the set of positive integers},\\
\mathbb{C}\;&:\ \text{the field of complex numbers},\\
p\;&:\ \text{an odd rational prime},\\
\mathbb{Z}_p\;&:\ \text{the ring of $p$-adic integers},\\
\mathbb{Q}_p\;&:\ \text{the field of fractions of $\mathbb{Z}_p$},\\
\mathbb{C}_p\;&:\ \text{the completion of a fixed algebraic closure of $\mathbb{Q}_p$}.
\end{aligned}
\]
We also write $\mathbb{N}_0=\mathbb{N}\cup\{0\}$.

Spectral zeta functions provide a standard tool for regularizing
determinants of differential operators in quantum physics. For a
quantum system with discrete energy spectrum
$\{\lambda_n\}_{n=0}^{\infty}$, where $\lambda_n\neq0$, the classical
spectral zeta function is initially defined by
\begin{equation}\label{eq:classical-zeta}
Z(s,\lambda)
=
\sum_{n=0}^{\infty}\frac{1}{(\lambda_n+\lambda)^s},
\end{equation}
for suitable $\lambda$ and $s$ in the domain of convergence. If
$Z(s,\lambda)$ admits a meromorphic continuation that is regular at
$s=0$, the associated functional determinant is defined by
\begin{equation}\label{eq:classical-det}
\log D(\lambda)
=
-\frac{\partial}{\partial s}Z(s,\lambda)\bigg|_{s=0}.
\end{equation}
Developed in the work of Hawking \cite{Hawking1977} and Voros
\cite{Voros1987,Voros1992}, this formalism has become an important
technique in quantum field theory, string theory, and cosmology.
Its properties and applications have been investigated by many
authors, including Freitas \cite{Freitas}, Zhang, Li, Liu and Dai
\cite{Zhang}, Fedosova, Rowlett and Zhang \cite{Fedosova}, Reyes Bustos
and Wakayama \cite{Bustos}, Kimoto and Wakayama \cite{KW}, and Cunha
and Freitas \cite{CF}. For a comprehensive account of its physical
applications, see Elizalde \cite{Elizalde}. Recent developments include
the localization of spectral zeta functions and determinants on Lie
groups \cite{Spreafico}, regularization techniques for functional
determinants in quantum field theory \cite{ShojiYamaguchi2025}, and
applications of supersymmetric zeta functions and determinants to
supersymmetric indices and Casimir energies
\cite{NakayamaOkazaki2026}.

For continuous spectra, this construction requires additional care.
In the classical archimedean setting, the corresponding formal
expressions are
\[
\zeta_A(s)
=
\operatorname{Tr}(A^{-s})
=
\int_{\sigma(A)}\lambda^{-s}\,\d m_A(\lambda),
\qquad
\log\det A=-\zeta_A'(0),
\]
where $m_A$ denotes an appropriate scalar spectral measure.
For operators with continuous spectrum, the ordinary trace need not
exist, and a relative or regularized formulation may be necessary.
When a spectral density
$\d m_A(\lambda)=\rho(\lambda)\,\d\lambda$ and a suitable heat trace
are available, one writes
\[
K(t)
=
\operatorname{Tr}(e^{-tA})
=
\int_0^\infty e^{-t\lambda}\rho(\lambda)\,\d\lambda.
\]
Under the relevant convergence assumptions, the Mellin transform
relates the heat trace to the spectral zeta function:
\[
\zeta_A(s)
=
\frac{1}{\Gamma(s)}
\int_0^\infty t^{s-1}K(t)\,\d t.
\]
The small-$t$ asymptotic expansion of $K(t)$ then governs the
meromorphic structure of $\zeta_A(s)$, including its poles and
residues. This approach, rooted in the work of Seeley \cite{Seeley}
and Gilkey \cite{Gilkey}, is a cornerstone of spectral geometry.
Alternatively, the spectral measure can be recovered from the
boundary values of the resolvent $R(z)=(A-zI)^{-1}$ on opposite sides
of the real axis, after which $\zeta_A(s)$ is obtained by integrating
$\lambda^{-s}$ against that measure. In their usual form, these
methods rely on archimedean analytic structures, such as real
heat-time evolution and boundary limits across the real axis.
They therefore do not transfer directly to the
$\mathbb{C}_p$-valued setting considered here, where a different
analytic construction is needed.

Despite these analytic differences, the real and $p$-adic fields
occupy parallel positions in number theory. By Ostrowski's theorem
\cite[Chapter~1]{Koblitz}, every nontrivial absolute value on
$\mathbb{Q}$ is equivalent either to the archimedean absolute value
$|\cdot|_\infty$ or to a $p$-adic absolute value $|\cdot|_p$ for some
prime $p$. The corresponding completions are thus precisely
$\mathbb{R}$ and the fields $\mathbb{Q}_p$. This classification
underlies the adelic viewpoint, in which arithmetic structures are
studied simultaneously over all completions of $\mathbb{Q}$.
From this perspective, the classical spectral zeta functions of
Hawking, Voros, and other mathematical physicists naturally suggest
the search for analogous constructions over the non-archimedean
fields $\mathbb{Q}_p$.

The development of $p$-adic analysis and mathematical physics
provides a broader context for this question. Vladimirov, Volovich
and Zelenov \cite{VladimirovVolovichZelenov1994} systematically
developed a framework encompassing pseudo-differential operators,
spectral theory, and $p$-adic quantum mechanics with complex-valued
wave functions. In the $p$-adic-valued setting, Khrennikov and
collaborators developed a parallel theory of distributions, Gaussian
integration, and quantum mechanics \cite{Khrennikov1994book}; see
also \cite{Albeverio1996,Khrennikov1991,Khrennikov1992,Albeverio1995}.
More recent work has established operator-algebraic foundations,
including trace-class operators and states in $p$-adic quantum
mechanics \cite{Aniello2023}. The construction of tensor products
of $p$-adic Hilbert spaces \cite{Aniello} further provides a basis
for the study of composite systems and entanglement. Another
important development is the $p$-adic AdS/CFT correspondence
\cite{Gubser2017}; see also \cite{Bhattacharyya2018}. Such models
incorporate non-archimedean structures into descriptions of
space-time or bulk geometry and call for analytic tools adapted to
those structures. Together, these developments provide further
motivation for the study of $p$-adic spectral zeta functions.

There is also a more concrete analytic motivation. In a complete
non-archimedean field, a series $\sum_{n=1}^{\infty}a_n$ converges
if and only if $a_n\to0$. This convergence criterion can make
$p$-adic analogues of classical expressions more tractable.
In \cite{HK2021}, motivated by a problem of Hilbert
\cite{HilbertProb}, we obtained an infinite-order linear
differential equation satisfied by $\zeta_{p,E}(s,\lambda)$ whose
defining series converges in a suitable region of $\mathbb{C}_p$.
Similarly, in their investigation of the quantum Rabi model, Kimoto
and Wakayama \cite{KW} derived divergent series expressions for the
values $\zeta(n,\lambda)$ of the Hurwitz zeta function at positive
integers \cite[(38) and (40)]{KW}, while proving convergence of the
corresponding series for the $p$-adic Hurwitz zeta function
$\zeta_p(s,\lambda)$ \cite[(50) and Remark~7.3]{KW}. These examples
show that $p$-adic analogues can exhibit substantially different
convergence behavior, reinforcing the interest in a $p$-adic theory
of spectral zeta functions.

To specify the analytic setting of the present paper, we distinguish
between two types of $p$-adic analysis. The first studies
complex-valued functions on $p$-adic spaces, such as functions from
$\mathbb{Q}_p$ to $\mathbb{C}$; the second studies $p$-adic-valued
functions, in particular functions from $\mathbb{C}_p$ to
$\mathbb{C}_p$. We work in the second setting. As observed by Vishik
\cite{Vishik1985}, this setting lacks a general notion of a Laplace
operator analogous to the classical one, which makes it difficult
to obtain spectral data directly from Schr\"odinger-type equations;
see also the opening paragraph of \cite{Kochubei}.

In our previous paper \cite{HuKim}, we avoided this difficulty by
taking discrete spectral data as the starting point. Let
$D_1=\{a\in\mathbb{C}_p:|a|_p\le1\}$ be the closed unit disc.
Given a discrete spectrum $\{\lambda_n\}_{n=0}^{\infty}$ of a
Hamiltonian, suppose that there is a bounded, locally analytic
interpolation function $f:D_1\to\mathbb{C}_p$ satisfying
$f(n)=\lambda_n$ for every $n\in\mathbb{N}_0$. We defined the
bosonic $p$-adic spectral zeta function by
\begin{equation}\label{eq:HK-bosonic}
\zeta_p^f(s,\lambda)
=
\frac{1}{s-1}
\int_{\mathbb{Z}_p}
\langle\lambda+f(a)\rangle^{1-s}\,\d a,
\end{equation}
and its fermionic counterpart by
\begin{equation}\label{eq:HK-fermionic}
\zeta_{p,E}^f(s,\lambda)
=
\int_{\mathbb{Z}_p}
\langle\lambda+f(a)\rangle^{1-s}\,\d\mu_{-1}(a).
\end{equation}
Here $\langle\cdot\rangle$ denotes the projection function,
integration with respect to $\d a$ is the Volkenborn integral
associated with the Haar distribution, and $\mu_{-1}$ is the bounded
$p$-adic measure recalled in Definition~\ref{def:mu-1}.
These constructions extend the classical $p$-adic Hurwitz and
Hurwitz-type Euler zeta functions. They apply, for example, to the
harmonic oscillator spectrum $\lambda_n=n+\frac12$ and the infinite
square well spectrum $\lambda_n=n^2$.

Although motivated by discrete spectra, this framework also
accommodates certain continuously parameterized families of spectral
data. A locally analytic function on $D_1$, such as $f(a)=a^2$ or
$f(a)=a+\frac12$, defines a continuous family of values and can
therefore describe more than the original sequence
$\{f(n)\}_{n\in\mathbb{N}_0}$. Nevertheless, the construction remains
tied to interpolation over $\mathbb{Z}_p$: both
\eqref{eq:HK-bosonic} and \eqref{eq:HK-fermionic} are integrals over
that parameter space, and the density of $\mathbb{N}_0$ in
$\mathbb{Z}_p$ means that the values at the nonnegative integers
determine the continuous restriction of $f$ to $\mathbb{Z}_p$.
Moreover, the integrating distributions are prescribed in advance.
In particular, the Haar distribution in the bosonic case is chosen
on the parameter space rather than obtained from the operator's
spectral data.

The essential limitation is thus not merely the choice of a
parameter space, but the requirement that the spectral data and
their distribution arise from a prescribed interpolation scheme.
A compact spectral set $\sigma\subset\mathbb{C}_p$ need not admit
the required locally analytic parameterization, and its spectral
distribution need not be induced by the Haar distribution.
The position operator in $p$-adic quantum mechanics on a compact
subset $\sigma\subset\mathbb{C}_p$ illustrates this distinction.
Its spectrum is $\sigma$ itself, so a nonempty choice of $\sigma$
without isolated points gives a non-discrete compact spectral set.
Its scalar spectral data may be specified by an arbitrary bounded
measure on $\sigma$, with no prescribed relation to the Haar
distribution on $\mathbb{Z}_p$. We return to this example in
Remark~\ref{rem:why-not-HK} at the end of
Section~\ref{sec:applications}.

These considerations lead to the main objective of this paper:
to construct $p$-adic spectral zeta functions directly from
resolvent data, without requiring interpolation over $\mathbb{Z}_p$.
We replace the interpolation function $f$ by the resolvent
$R(z)=(zI-T)^{-1}$ and use Vishik's inverse Stieltjes transform
to recover a generalized spectral distribution. This transform
associates a unique element $\mu\in\cL^*(\sigma)$ with each Krasner
analytic function on $\mathbb{C}_p\setminus\sigma$ that vanishes at
infinity. It therefore allows us to work with arbitrary compact
spectral sets and with generalized distributions that need not be
bounded measures.

The analytic foundation of the inverse transform is the Shnirelman
integral. Introduced by Shnirelman in 1938 \cite{Shnirelman}, this
integral serves as a $p$-adic analogue of contour integration.
Since $p$-adic spaces are totally disconnected, it is defined through
averages over roots of unity on a $p$-adic circle rather than along
a continuous contour. The resulting theory includes $p$-adic
versions of the Cauchy integral formula, the residue theorem, and
the maximum modulus principle. We recall the relevant facts in
Subsection~\ref{sec:shnirelman}; for a detailed treatment, see
\cite[Appendix, \S3]{Koblitz2}.

Together with the distribution theory developed by Amice--V\'elu
\cite{AmiceVelu} and Vishik \cite{Vishik1985}, the Shnirelman
integral provides the means to pass from a scalar resolvent to a
generalized distribution. More precisely, let $T$ be a bounded
linear operator on a $\mathbb{C}_p$-Banach space. Its resolvent
\[
R(z)=(zI-T)^{-1}
\]
is analytic outside the spectrum. Under the admissibility
assumptions of Definition~\ref{def:admissible-operator}, a chosen
scalarization $F_T$ of $R$ belongs to
$H_0(\mathbb{C}_p\setminus\sigma)$, the space of Krasner analytic
functions vanishing at infinity on the complement of a suitable
compact set $\sigma$. Vishik's inverse Stieltjes theorem
(Theorem~\ref{thm:invStieltjes}) then yields a unique
$\mu\in\cL^*(\sigma)$ whose Stieltjes transform is $F_T$.
Here $\cL^*(\sigma)$ denotes the continuous dual of the space
$\cL(\sigma)$ of locally analytic functions on $\sigma$; see
Definition~\ref{def:generalized-distribution}. We call $\mu$ the
$p$-adic spectral distribution of $T$, relative to the chosen
scalarization (Definition~\ref{def:spectral-distribution}).
This distribution, rather than an interpolation function, is the
basic object in our construction.

For the corresponding spectral distributions, we define the bosonic
and fermionic $p$-adic spectral zeta functions by
\begin{equation}\label{eq:bosonic-def}
\zeta_p^\mu(s,\lambda)
=
\frac{1}{s-1}
\int_\sigma
\langle\lambda+x\rangle^{1-s}\,\d\mu(x)
\end{equation}
and
\begin{equation}\label{eq:fermionic-def}
\zeta_{p,E}^\nu(s,\lambda)
=
\int_\sigma
\langle\lambda+x\rangle^{1-s}\,\d\nu(x),
\end{equation}
respectively, with the hypotheses and domains specified in
Definitions~\ref{def:bosonic-zeta} and~\ref{def:fermionic-zeta}.
These formulas retain the form of \eqref{eq:HK-bosonic} and
\eqref{eq:HK-fermionic}, but replace integration over the parameter
space by the action of a distribution on the spectral set.
We establish $C^\infty$-dependence on $s$ and local analyticity in
$\lambda$ in Theorems~\ref{thm:welldef}
and~\ref{thm:fermionic-analyticity}. We also express the values at
non-positive integers in terms of generalized Bernoulli and Euler
functionals (Theorems~\ref{thm:special-bosonic}
and~\ref{thm:special-fermionic}; see
Definition~\ref{def:Bernoulli} and~(\ref{eq:Enu})).
The compatibility theorem (Theorem~\ref{thm:compat}) shows that,
when the discrete spectral data admit a locally analytic
interpolation, the zeta functions of \cite{HuKim} are recovered
from the corresponding inverse Stieltjes distributions.

Although the bosonic and fermionic constructions are parallel in
form, their compatibility proofs differ in an important analytic
respect. The discrete bosonic construction uses the unbounded Haar
distribution, for which continuity alone does not guarantee
Volkenborn integrability. Its compatibility proof therefore requires
estimates in differentiable function spaces, together with the
indefinite sum and suitable Taylor expansions; see the proof of
Theorem~\ref{thm:compat}. By contrast, the bounded measure
$\mu_{-1}$ integrates every continuous function on $\mathbb{Z}_p$.
The fermionic compatibility argument consequently requires only
continuity of the parameterization and is substantially simpler.

The associated $p$-adic functional determinants are introduced in
Definition~\ref{def:functional-det}. In the classical theory, the
zeta-regularized determinant $\log\det A=-\zeta_A'(0)$ enters the
one-loop effective action and plays a central role in quantum field
theory, spectral geometry, and the computation of Casimir energies
\cite{Elizalde,Hawking1977,Kirsten,Voros1992}. Following the
normalization of \cite{HuKim}, we define $p$-adic analogues through
normalized derivatives of the spectral zeta functions at $s=0$.
We prove integral representations in
Theorem~\ref{thm:integral-rep} and a Stirling-type expansion in
Theorem~\ref{thm:stirling}. In the region where this expansion
holds, the determinant is determined by the moments of the spectral
distribution. These moments are precisely the coefficients in the
expansion of the scalar resolvent at infinity. Thus the resolvent
expansion plays a role analogous to that of heat kernel
coefficients in classical spectral regularization.

Our principal application is the position operator in $p$-adic
quantum mechanics on an arbitrary compact subset
$\sigma\subset\mathbb{C}_p$. Its spectrum is exactly $\sigma$, and
its explicit resolvent allows us to identify the spectral
distribution and derive formulas for the associated zeta functions
and functional determinant. A natural example is the closed unit
ball of a finite extension $K/\mathbb{Q}_p$, which represents a
bounded spatial region and gives a non-discrete compact spectral
set. This application illustrates the treatment of continuous
spectral data within the inverse Stieltjes framework, while the
compatibility theorem in Subsection~\ref{Compatibility} establishes
its connection with our earlier interpolation-based construction.

The paper is organized as follows. Section~\ref{sec:prelim} recalls
the necessary background on $p$-adic numbers, distributions, and
measures, together with the Shnirelman integral and the inverse
Stieltjes theorem. In
Subsection~\ref{sec:spectral-distribution}, we construct the
$p$-adic spectral distribution from a scalar resolvent and establish
the corresponding moment expansion at infinity
(Proposition~\ref{prop:expansion}). Section~\ref{sec:zeta} introduces
the bosonic and fermionic spectral zeta functions, establishes their
regularity properties and special value formulas, and proves
compatibility with \cite{HuKim}. Section~\ref{sec:det} develops the
associated functional determinants, including their integral
representations and the Stirling expansion. Finally,
Section~\ref{sec:applications} applies the theory to the position
operator in $p$-adic quantum mechanics on a compact subset of
$\mathbb{C}_p$.

\section{Preliminaries}\label{sec:prelim}

This section collects the background material used throughout the paper. In subsection~\ref{sec:p-adic-numbers}, we recall the multiplicative decomposition of \(p\)-adic numbers and the projection function. In subsection~\ref{sec:measures}, we introduce \(p\)-adic distributions and measures on \(\Zp\), including the Haar distribution and the fermionic measure \(\mu_{-1}\). In subsection~\ref{sec:shnirelman}, we develop the Shnirelman integral and the inverse Stieltjes transform. Finally, in subsection~\ref{sec:spectral-distribution}, we apply these tools to associate a \(p\)-adic spectral distribution to an admissible operator via its resolvent, which is the key construction for defining spectral zeta functions in the continuous-spectrum setting.

\subsection{Basic facts about $p$-adic numbers}
\label{sec:p-adic-numbers}

Let \(p\) be an odd prime, and let \(\Qp\) be the completion of
\(\Q\) with respect to the \(p\)-adic absolute value
\(|\cdot|_p\), normalized by
\[
|p|_p=p^{-1}.
\]
Let \(\overline{\Q}_p\) be a fixed algebraic closure of \(\Qp\),
and let \(\Cp\) be its completion. Throughout the paper, we work
in the setting sometimes called the ``second kind'' of
\(p\)-adic analysis, in which the functions, distributions, and
integrals under consideration take values in \(\Cp\).

To recall the definition of the absolute value, write a nonzero
rational number \(x\) as
\[
x=p^n\frac{a}{b},
\]
where \(n\in\Z\), \(a,b\in\Z\) are prime to \(p\), and \(a/b\)
is in lowest terms with \(b>0\). This representation is unique,
and
\[
|x|_p=p^{-n}.
\]
The absolute value extends to \(\Qp\) and \(\Cp\) and satisfies
the ultrametric, or strong triangle, inequality
\[
|x+y|_p\le\max\{|x|_p,|y|_p\}.
\]
This inequality accounts for several basic geometric properties:
every point of a \(p\)-adic disc is a center of that disc, and
two discs are either disjoint or one contains the other.
In particular, both \(\Qp\) and \(\Cp\) are totally disconnected.

We next recall the multiplicative decomposition of nonzero
\(p\)-adic numbers and the associated projection function
\(\langle\cdot\rangle\). These will be used to define the
\(p\)-adic powers appearing in our spectral zeta functions.
Our exposition follows \cite[Sec.~2]{TP2011}.

We begin with units in \(\Zp\). For each \(a\in\Zp^\times\),
there is a unique \((p-1)\)st root of unity
\(\omega(a)\in\Zp^\times\) such that
\[
a\equiv\omega(a)\pmod p.
\]
The map \(a\mapsto\omega(a)\) is the Teichm\"uller character.
Dividing \(a\) by its Teichm\"uller representative gives the
projection
\[
\langle a\rangle=\omega(a)^{-1}a.
\]
Thus
\[
a=\omega(a)\langle a\rangle,
\qquad
\langle a\rangle\in1+p\Zp.
\]
This decomposition separates the root-of-unity factor from the
principal-unit factor.

To extend this construction to \(\Cp^\times\), let
\[
\mathfrak o_{\Cp}
=
\{x\in\Cp:|x|_p\le1\},
\qquad
\mathfrak m_{\Cp}
=
\{x\in\Cp:|x|_p<1\}
\]
denote the valuation ring of \(\Cp\) and its maximal ideal,
respectively. For \(x\in\Cp^\times\), write \(v_p(x)\in\Q\)
for the additive valuation normalized by \(v_p(p)=1\), so that
\[
|x|_p=p^{-v_p(x)}.
\]
Choose once and for all a compatible system of fractional powers
\(p^r\in\Cp^\times\), \(r\in\Q\), satisfying
\[
p^{r+t}=p^r p^t
\quad(r,t\in\Q),
\]
and agreeing with the usual powers when \(r\in\Z\).

For \(a\in\Cp^\times\), put
\[
u=p^{-v_p(a)}a,
\]
so that \(|u|_p=1\). There is a unique root of unity
\(\hat a\) of order prime to \(p\) satisfying
\[
|\hat a-u|_p<1.
\]
It is the Teichm\"uller representative of the residue class of
\(u\), and it can also be obtained from the limit
\[
\hat a=\lim_{n\to\infty}u^{p^{n!}}.
\]
We then define
\[
\langle a\rangle
=
p^{-v_p(a)}\frac{a}{\hat a}.
\]
By construction,
\[
a=p^{v_p(a)}\hat a\langle a\rangle,
\qquad
\langle a\rangle\in1+\mathfrak m_{\Cp}.
\]

Equivalently, this gives a multiplicative decomposition
\[
\Cp^\times\simeq p^{\Q}\times\mu\times D,
\]
where
\[
p^{\Q}=\{p^r:r\in\Q\},
\]
\(\mu\) is the group of roots of unity of order prime to \(p\),
and
\[
D=1+\mathfrak m_{\Cp}
=
\{x\in\Cp:|x-1|_p<1\}.
\]
The decomposition depends on the chosen system of fractional
powers of \(p\). Once that choice is fixed, all three factors
are uniquely determined. Changing the choice can modify the
factors by roots of unity.

It is convenient to combine the valuation and root-of-unity
factors into the extended Teichm\"uller character
\[
\omega_v(a)
=
\frac{a}{\langle a\rangle}
=
p^{v_p(a)}\hat a,
\qquad a\in\Cp^\times.
\]
Thus the decomposition may also be written as
\[
a=\omega_v(a)\langle a\rangle.
\]
For \(a\in\Qp^\times\), the valuation \(v_p(a)\) is an integer,
the representative \(\hat a\) is a \((p-1)\)st root of unity,
and
\[
\langle a\rangle\in1+p\Zp.
\]
Hence the construction on \(\Cp^\times\) extends the familiar
decomposition on \(\Qp^\times\), and we use the same notation
in both settings.

We now record the local behavior of these factors. The maps
\begin{equation}\label{maps}
a\longmapsto p^{v_p(a)},
\qquad
a\longmapsto\hat a,
\qquad
a\longmapsto\omega_v(a)
\end{equation}
are locally constant on \(\Cp^\times\). Indeed, if
\(|b-a|_p<|a|_p\), then \(v_p(b)=v_p(a)\), and the normalized
units \(p^{-v_p(a)}b\) and \(p^{-v_p(a)}a\) have the same
residue class. Their Teichm\"uller representatives therefore
coincide. In particular, the derivatives of all three maps
in \eqref{maps} vanish.

By contrast, the projection \(a\mapsto\langle a\rangle\) is
locally analytic and is not locally constant. Since
\[
\langle a\rangle=\frac{a}{\omega_v(a)}
\]
and the denominator is locally constant, we obtain
\begin{equation}\label{deri}
\frac{\d}{\d a}\langle a\rangle
=
\frac{1}{\omega_v(a)}
=
\frac{\langle a\rangle}{a},
\qquad a\in\Cp^\times.
\end{equation}
This identity will be used repeatedly when differentiating
expressions involving the projection function.

Finally, for \(a\in\Cp^\times\) and \(s\in\Zp\), we define
the \(p\)-adic power \(\langle a\rangle^s\) by the binomial
series
\[
\langle a\rangle^s
=
\sum_{n=0}^{\infty}
\binom{s}{n}\bigl(\langle a\rangle-1\bigr)^n.
\]
Since \(\binom{s}{n}\in\Zp\) and
\(|\langle a\rangle-1|_p<1\), the terms tend to zero, and
the series converges in \(\Cp\). For fixed \(a\), the resulting
function is locally analytic in \(s\in\Zp\); for fixed \(s\),
it is locally analytic in \(a\in\Cp^\times\).
These properties provide the basic analytic input for the
kernels \(\langle\lambda+x\rangle^{1-s}\) used below.

\subsection{$p$-adic measures and integrals on $\mathbb{Z}_p$}
\label{sec:measures}

Here we introduce the notions of \(p\)-adic distributions and bounded measures on \(\Zp\), and we recall the Haar distribution and the fermionic measure \(\mu_{-1}\). These measures will be used to define the bosonic and fermionic spectral zeta functions in the next section. The fermionic measure is particularly important for our purposes, since it gives rise to the Euler-type zeta functions studied in the paper.

\begin{definition}[Distribution and bounded measure]
\label{def:distribution-on-Zp}
A \(p\)-adic \emph{distribution} on \(\Zp\) is a finitely additive map
\[
\mu:\{\text{compact open subsets of }\Zp\}\longrightarrow\Cp.
\]
If there exists a constant \(C>0\) such that
\[
|\mu(U)|_p\le C
\]
for every compact open subset \(U\subset\Zp\), then \(\mu\) is called a \emph{bounded measure}.
\end{definition}

For a locally analytic function \(f:\Zp\to\Cp\) and a distribution \(\mu\) on \(\Zp\), the integral of \(f\) against \(\mu\) is defined by the limit of Riemann sums over residue classes:
\[
\int_{\Zp} f(x)\,\d\mu(x)
=
\lim_{N\to\infty}\sum_{a=0}^{p^N-1} f(a)\,\mu\bigl(a+p^N\Zp\bigr),
\]
whenever the limit exists in \(\Cp\).

\begin{definition}[Haar distribution and Volkenborn integral]
\label{def:haar}
The \emph{Haar distribution} on \(\Zp\) is defined by
\[
\mu_{\mathrm{Haar}}\bigl(a+p^N\Zp\bigr)=\frac{1}{p^N}.
\]
It is unbounded in general. The associated integral is the \emph{Volkenborn integral}, written as
\[
\int_{\Zp} f(a)\,\d a
=
\lim_{N\to\infty}\frac{1}{p^N}\sum_{a=0}^{p^N-1}f(a),
\]
for locally analytic \(f:\Zp\to\Cp\) whenever the limit exists.
\end{definition}

\begin{definition}[{\(\mu_{-1}\) and the fermionic integral}]
\label{def:mu-1}
Assume \(p\) is odd. The measure \(\mu_{-1}\) is defined on residue classes by
\[
\mu_{-1}\bigl(a+p^N\Zp\bigr)=(-1)^a,
\]
for \(0\le a<p^N\). It is a bounded measure. The associated integral is
\[
\int_{\Zp} f(a)\,\d\mu_{-1}(a)
=
\lim_{N\to\infty}\sum_{a=0}^{p^N-1}f(a)(-1)^a.
\]
\end{definition}

\begin{remark}
This measure was independently found by Katz \cite[p.~486]{Katz} (in Katz's notation, the \(\mu^{(2)}\)-measure), Shiratani and Yamamoto \cite{Shi}, Osipov \cite{Osipov}, Lang \cite{Lang} (in Lang's notation, the \(E_{1,2}\)-measure), T. Kim \cite{TK} from very different viewpoints.
Obviously, in contrast with the Haar distribution, the \(\mu_{-1}\)-measure is bounded under the \(p\)-adic valuation, so it can be applied to integrate the continuous functions on \(\mathbb{Z}_{p}\)
(see \cite[p.~39, Theorem 6]{Koblitz}).
\end{remark}

It is interesting to note that, although \(\mu_{-1}\) assigns the alternating signs \((-1)^a\) to residue classes, its total mass is still \(1\). This is a direct consequence of the fact that \(p\) is odd.

\begin{proposition}[{Total mass under $\mu_{-1}$}]
\label{prop:mu-1-total-mass}
The measure \(\mu_{-1}\) is normalized so that
\[
\int_{\Zp} \d\mu_{-1}(a) = 1.
\]
\end{proposition}

\begin{proof}
For every \(N\ge 0\), the sets
\[
a+p^N\Zp,\quad a=0,1,\dots,p^N-1,
\]
form a disjoint clopen cover of \(\Zp\). Therefore, by finite additivity,
\[
\int_{\Zp} \d\mu_{-1}(a)
=
\mu_{-1}(\Zp)
=
\sum_{a=0}^{p^N-1} \mu_{-1}(a+p^N\Zp)
=
\sum_{a=0}^{p^N-1} (-1)^a.
\]
Since \(p\) is odd, \(p^N\) is odd, and the alternating sum
\[
\sum_{a=0}^{p^N-1} (-1)^a = 1.
\]
Hence
\[
\int_{\Zp} \d\mu_{-1}(a)=1.
\]
Thus the proposition follows.
\end{proof}

These two measures are not only the building blocks of the bosonic and fermionic spectral zeta functions studied in this paper; they also lie at the heart of Iwasawa theory. To see this, recall that the classical \(p\)-adic Hurwitz zeta function \(\zeta_p(s,\lambda)\) is defined by the \(p\)-adic Mellin transform of the Haar distribution,
\[
\zeta_p(s,\lambda)
=
\frac{1}{s-1}\int_{\Zp}\langle \lambda+a\rangle^{1-s}\,\d a
\]
for $\lambda\in\Cp\setminus\Zp$ and the \(p\)-adic Hurwitz-type Euler zeta function \(\zeta_{p,E}(s,\lambda)\) is defined by the \(p\)-adic Mellin transform of the \(\mu_{-1}\)-measure,
\[
\zeta_{p,E}(s,\lambda)
=
\int_{\Zp}\langle \lambda+a\rangle^{1-s}\,\d\mu_{-1}(a)
\]
for $\lambda\in\Cp\setminus\Zp,$ see \cite[Definition 11.2.5]{Cohen} and \cite[Definition 3.3]{HK2012}, respectively. Both interpolate the corresponding classical zeta functions at non-positive integers:
\[
\zeta_p(1-m,\lambda)
=-\frac{1}{\omega_v^m(\lambda)}\frac{B_m(\lambda)}{m},
\quad
\zeta_{p,E}(1-m,\lambda)
=\frac{1}{\omega_v^m(\lambda)}E_m(\lambda),
\]
where \(B_m(\lambda)\) and \(E_m(\lambda)\) are the Bernoulli and Euler polynomials; see \cite[Proposition 11.2.6]{Cohen} and \cite[Theorem 3.8(2)]{HK2012}.

Using these zeta functions as building blocks, one defines the corresponding \(p\)-adic \(L\)-functions. The Kubota--Leopoldt \(p\)-adic \(L\)-function
\[
L_p(s,\chi)
=
\frac{1}{s-1}\int_{\Zp}\chi(a)\langle a\rangle^{1-s}\,\d a
\]
is associated with the ideal class group of the \(p^n\)-th cyclotomic field \(\Q(\zeta_{p^n})\); see \cite{Iw} and \cite{Wa}. On the other side, the \(p\)-adic \(L\)-function
\[
L_{p,E}(s,\chi)
=
\int_{\Zp}\chi(a)\langle a\rangle^{1-s}\,\d\mu_{-1}(a)
\]
is associated with the $(S,\{2\})$-refined ideal class group of $\Q(\zeta_{p^n})$; see \cite[Propositions 3.3 and 3.4]{HK2016} for details. In this sense, the Haar distribution and the \(\mu_{-1}\)-measure encode arithmetic information about cyclotomic fields, and the bosonic and fermionic spectral zeta functions introduced in this paper can be regarded as spectral analogues of the Kubota--Leopoldt and \((S,\{2\})\)-refined \(p\)-adic \(L\)-functions, respectively.

This arithmetic interpretation provides additional motivation for the
constructions in this paper: the same measures that govern the
interpolation of classical \(L\)-values in Iwasawa theory also govern
the regularization of spectral determinants in \(p\)-adic quantum
mechanics. Having introduced the two measures and their arithmetic
applications, we now turn to the analytic tools needed to construct
spectral distributions from resolvents.

\subsection{Shnirelman integral and inverse Stieltjes transform}
\label{sec:shnirelman}

Shnirelman integral is the \(p\)-adic analogue of the contour integral
in complex analysis. It was introduced by Shnirelman in 1938
\cite{Shnirelman} and provides \(p\)-adic versions of the residue
theorem and the Cauchy integral formula. Since \(p\)-adic spaces are
totally disconnected, the integral is not defined along a continuous
contour but by averaging over roots of unity on a \(p\)-adic circle.
This construction allows one to extract the constant term of a Laurent
expansion, playing the same role as the classical residue does in
complex analysis.

In this subsection we use this tool to relate
resolvent-type functions to \(p\)-adic distributions. We first recall
the space of Krasner analytic functions and the Shnirelman integral,
and then state the inverse Stieltjes transform, which gives an
isomorphism between the space of generalized distributions on a
compact set \(\sigma\) and the space of functions analytic on the
complement of \(\sigma\). This result is the main bridge between
spectral data and distribution theory in our setting. Our exposition
follows the treatment by Koblitz in \cite[Appendix, \S3]{Koblitz2}, to
which we refer for proofs and further details.

\begin{definition}[{Locally analytic and $\cL(\sigma)$,
see \cite[p.~136]{Koblitz2}}]
\label{def:locally-analytic}
Let \(\sigma\subset\Cp\) be compact. A function \(f:\sigma\to\Cp\) is
called \emph{locally analytic} if for every \(a\in\sigma\), there is
an open neighbourhood \(U\subset\Cp\) of \(a\) such that \(f\) is
given by a convergent power series on \(\sigma\cap U\). The space of
all locally analytic functions on \(\sigma\) is denoted by
\(\cL(\sigma)\).
\end{definition}

\begin{definition}[{$\mathbb D_\sigma(r)$ and $B_r(\sigma)$,
see \cite[Appendix, \S3, p.~136]{Koblitz2}}]
\label{def:Br}
The space \(\cL(\sigma)\) is equipped with the inductive limit topology
of the Banach spaces \(B_r(\sigma)\), \(r>0\), defined as follows. For
each \(r>0\), let
\begin{equation}\label{Dar}
\mathbb D_\sigma(r)=\bigcup_{a\in\sigma}D_a(r),
\quad
D_a(r)=\{z\in\Cp:|z-a|_p<r\},
\end{equation}
and let \(B_r(\sigma)\) be the space of functions on
\(\mathbb D_\sigma(r)\) that are given by a convergent power series on
each disc \(D_a(r)\). Equipped with the norm
\[
\|f\|_r=\max_{z\in\mathbb D_\sigma(r)}|f(z)|_p,
\]
the space \(B_r(\sigma)\) is a Banach space over \(\Cp\). The space
\(\cL(\sigma)\) is the union of the \(B_r(\sigma)\), and its topology
is the inductive limit topology.
\end{definition}

\begin{definition}[{Generalized distribution,
see \cite[p.~136, Lemma 7]{Koblitz2}}]
\label{def:generalized-distribution}
A \emph{generalized distribution} on \(\sigma\) is a continuous
\(\Cp\)-linear map
\[
\mu:\cL(\sigma)\longrightarrow\Cp.
\]
The space of all such functionals is denoted by \(\cL^*(\sigma)\). For
\(\mu\in\cL^*(\sigma)\) and \(r>0\), the norm of \(\mu\) on
\(B_r(\sigma)\) is
\begin{equation}\label{norm}
\|\mu\|_r=\sup_{0\neq f\in B_r(\sigma)}
\frac{|\mu(f)|_p}{\|f\|_r}.
\end{equation}
\end{definition}

Every bounded \(p\)-adic measure on \(\sigma\) defines an element of
\(\cL^*(\sigma)\) by integration, namely the continuous linear
functional
\[
f\longmapsto \mu(f):=\int_\sigma f\,\d\mu.
\]
Thus the notation \(\mu(f)\) denotes the action of this functional on
\(f\), which coincides with the integral of \(f\) against the bounded
measure \(\mu\). However, \(\cL^*(\sigma)\) also contains many
genuinely unbounded functionals that are not represented by measures;
see \cite[p.~136, Lemma 7]{Koblitz2} for the precise criterion
distinguishing measures within \(\cL^*(\sigma)\).

Let
\[
\Omega_\sigma=\Cp\setminus\sigma
\]
be the complement of \(\sigma\) in \(\Cp\). For \(r>0\), write
\[
\Omega_\sigma(r)=\{z\in\Cp:\operatorname{dist}(z,\sigma)\ge r\},
\]
where \(\operatorname{dist}(z,\sigma)=\inf_{x\in\sigma}|z-x|_p\).

\begin{definition}[{Krasner analytic and $H_0(\Omega_\sigma)$,
see \cite[pp.~133--134]{Koblitz2}}]
\label{def:krasner-analytic}
A function \(F:\Omega_\sigma\to\Cp\) is called \emph{Krasner
analytic} on \(\Omega_\sigma\) if, for every \(r>0\), the restriction
\(F|_{\Omega_\sigma(r)}\) is the uniform limit of rational functions
whose poles all lie in \(\sigma\).

The space \(H_0(\Omega_\sigma)\) consists of all Krasner analytic
functions \(F\) on \(\Omega_\sigma\) that vanish at infinity, i.e.
\[
\lim_{|z|_p\to\infty}F(z)=0.
\]
It is equipped with the locally convex topology defined by the seminorms
\[
\|F\|_{\Omega,r}:=\sup_{z\in\Omega_\sigma(r)}|F(z)|_p,\quad r>0,
\]
where the subscript \(\Omega\) distinguishes this seminorm from the
norm \(\|\cdot\|_r\) on \(B_r(\sigma)\).
\end{definition}

Now let \(a\in\Cp\) and let \(\Gamma\in\Cp^\times\). Put
\(r=|\Gamma|_p\). Consider the circle
\[
C(a,r)=\{x\in\Cp:|x-a|_p=r\}.
\]

In 1938, Shnirelman \cite{Shnirelman} introduced the following integral, which is the \(p\)-adic analogue of the contour integral
in complex analysis.
\begin{definition}[{Shnirelman integral,
see \cite[p.~129, Definition]{Koblitz2}}]
\label{def:shnirelman}
Let \(f\) be a function admitting a convergent Laurent expansion
\[
f(x)=\sum_{k\in\Z}c_k\,(x-a)^k
\]
on an annulus containing the circle \(C(a,r)\). The \emph{Shnirelman
integral} of \(f\) around \(a\) is defined by
\[
\int_{a,\Gamma}f(x)\,\d x
=
\lim_{\substack{n\to\infty\\ p\nmid n}}
\frac1n\sum_{\xi^n=1}f(a+\xi\Gamma),
\]
whenever the right-hand side converges in \(\Cp\).
\end{definition}

\begin{proposition}[{see \cite[pp.~129--130, Lemma 1(3)]{Koblitz2}}]
\label{prop:shnirelman-evaluation}
If
\[
f(x)=\sum_{k\in\Z}c_k\,(x-a)^k
\]
converges uniformly on the circle \(|x-a|_p=r\), then
\[
\int_{a,\Gamma}f(x)\,\d x=c_0.
\]
Equivalently,
\[
\int_{a,\Gamma}(x-a)^k\,\d x=
\begin{cases}
1,&k=0,\\
0,&k\ne0.
\end{cases}
\]
\end{proposition}
\begin{proof}
See \cite[pp.~129--130, Lemma 1(3)]{Koblitz2}.
\end{proof}

\begin{proposition}[{Uniform convergence and norm estimate,
see \cite[p.~129, Lemma~1(1) and (2)]{Koblitz2}}]
\label{prop:shnirelman-uniform}
Let \(a\in\Cp\), \(\Gamma\in\Cp^\times\), and \(r=|\Gamma|_p\). Let
\(f_n,f\) be functions defined on the circle
\(C(a,r)=\{x\in\Cp:|x-a|_p=r\}\).

\begin{enumerate}[(1)]
\item If \(f_n\to f\) uniformly on \(C(a,r)\), and if
\(\int_{a,\Gamma}f_n(x)\,\d x\) exists for each \(n\), then
\(\int_{a,\Gamma}f(x)\,\d x\) exists and
\[
\int_{a,\Gamma}f_n(x)\,\d x
\longrightarrow
\int_{a,\Gamma}f(x)\,\d x
\quad(n\to\infty).
\]
\item If \(\int_{a,\Gamma}f(x)\,\d x\) exists, then
\[
\left|\int_{a,\Gamma}f(x)\,\d x\right|_p
\le
\max_{|x-a|_p=r}|f(x)|_p.
\]
\end{enumerate}
\end{proposition}

\begin{proof}
See \cite[p.~129, Lemma~1(1) and (2)]{Koblitz2}. Both statements follow from the
definition of the Shnirelman integral as a limit of finite averages,
together with the strong triangle inequality.
\end{proof}

We next introduce a notion that will simplify the statements of the
subsequent results.

\begin{definition}[{Admissible cover}]
\label{def:admissible-cover}
Let \(\sigma\subset\Cp\) be compact and let \(r>0\). An
\emph{admissible cover} of \(\sigma\) of radius \(r\) consists of a
finite cover
\begin{equation}\label{cover}
\sigma\subset\bigcup_{i}D_i,
\qquad
D_i=\{x\in\Cp:|x-a_i|_p<r\},
\end{equation}
where the discs \(D_i\) are pairwise disjoint clopen discs of radius
\(r\) with \(a_i\in\sigma\), together with a choice, for each \(i\),
of an element \(\Gamma_i\in\Cp^\times\) satisfying
\(|\Gamma_i|_p=r\). We write \(\{D_i,\Gamma_i\}_i\) for such an
admissible cover.
\end{definition}
The right-hand sides of the formulas in
Proposition~\ref{prop:cauchy} and Theorem~\ref{thm:invStieltjes} below
are independent of the particular choice of the admissible cover; see
\cite[pp.~138--142]{Koblitz2} for details.

\begin{proposition}[{$p$-adic Cauchy integral formula,
see \cite[p.~131, Lemma~4]{Koblitz2}}]
\label{prop:cauchy}
Let \(\sigma\subset\Cp\) be compact, let \(g\in\cL(\sigma)\), and let
\(\{D_i,\Gamma_i\}_i\) be an admissible cover of \(\sigma\) of some
radius \(r>0\) small enough so that \(g\) admits an analytic
continuation to a neighbourhood of \(\mathbb D_\sigma(r)\)
(see Definition~\ref{def:admissible-cover}). Then for every
\(z\in\sigma\),
\[
\sum_i\int_{a_i,\Gamma_i}
\frac{g(x)\,(x-a_i)}{x-z}\,\d x
=
g(z).
\]
More generally, for every \(m\ge0\),
\[
\sum_i\int_{a_i,\Gamma_i}
\frac{g(x)\,(x-a_i)}{(x-z)^{m+1}}\,\d x
=
\frac{1}{m!}g^{(m)}(z).
\]
\end{proposition}

\begin{proof}
See \cite[pp.~131--132, Lemma~4]{Koblitz2}. The proof reduces, by
linearity and continuity, to the case \(g(x)=(x-a)^n\) for some
\(a\in\sigma\) and \(n\ge0\), and then uses the Laurent expansion of
\((x-z)^{-m-1}\) together with
Proposition~\ref{prop:shnirelman-evaluation}.
\end{proof}

\begin{remark}
Thus the Shnirelman integral is the \(p\)-adic analogue of taking the
constant term in a Laurent expansion, rather than of the classical
residue coefficient \(c_{-1}\). For the \(p\)-adic versions of the
residue theorem and the maximum modulus principle we refer to
\cite[p.~132, Lemmas~5 and~6]{Koblitz2}.
\end{remark}

\begin{definition}[{Stieltjes transform,
see \cite[p.~137, Definition]{Koblitz2}}]
\label{def:stieltjes-transform}
Let \(\mu\in\cL^*(\sigma)\). For \(z\in\Omega_\sigma\), the function
\[
x\longmapsto\frac{1}{z-x}
\]
is locally analytic on \(\sigma\). The \emph{Stieltjes transform} of
\(\mu\) is the function \(S_\mu:\Omega_\sigma\to\Cp\) defined by
\[
S_\mu(z)
=\left(\mu(x),\frac{1}{z-x}\right),
\quad z\in\Omega_\sigma,
\]
where \((\mu,\varphi)\) denotes the action of \(\mu\) on
\(\varphi\in\cL(\sigma)\).
\end{definition}

The following theorem is the main bridge between resolvents and
spectral distributions.

\begin{theorem}[{Vishik, the inverse Stieltjes theorem,
see \cite[pp.~138--142, Theorem]{Koblitz2}}]
\label{thm:invStieltjes}
The map
\[
S:\cL^*(\sigma)\longrightarrow H_0(\Omega_\sigma),
\quad
\mu\longmapsto S_\mu,
\]
is an isomorphism of locally convex \(\Cp\)-vector spaces. Its inverse
\[
V:H_0(\Omega_\sigma)\longrightarrow\cL^*(\sigma)
\]
is described as follows. Let \(F\in H_0(\Omega_\sigma)\) and
\(f\in\cL(\sigma)\). Choose an admissible cover
\(\{D_i,\Gamma_i\}_i\) of \(\sigma\) of some radius \(r>0\) small
enough so that \(f\) admits an analytic continuation to a
neighbourhood of \(\{x\in\Cp:\operatorname{dist}(x,\sigma)\le r\}\)
(see Definition~\ref{def:admissible-cover}). Then
\begin{equation}\label{VF}
(VF)(f)
=
\sum_i
\int_{a_i,\Gamma_i}
F(x)\,f(x)\,(x-a_i)\,\d x.
\end{equation}
The right-hand side is independent of the choice of the radius \(r\),
of the admissible cover, and of the choices of \(\Gamma_i\).
\end{theorem}

\begin{proof}
See \cite[pp.~138--142, Theorem]{Koblitz2}. The proof proceeds by
verifying that \(S\) and \(V\) are continuous and mutually inverse; the
key ingredients are the Shnirelman integral estimates
\cite[p.~129, Lemma 1]{Koblitz2}, the \(p\)-adic Cauchy integral
formula \cite[p.~131, Lemma 4]{Koblitz2}, and the characterization of
measures within \(\cL^*(\sigma)\) \cite[p.~136, Lemma 7]{Koblitz2}.
\end{proof}

\begin{remark}[{see \cite[p.~136, Lemma 7 and p.~142,
Remark]{Koblitz2}}]
The space \(\cL^*(\sigma)\) is much larger than the space of bounded
measures on \(\sigma\). The subspace of measures is characterized by the
boundedness of the norm \(\|\mu\|_r\) as \(r\to0\). The intermediate
class of \(h\)-admissible measures, studied by Amice--V\'elu
\cite{AmiceVelu} and Vishik \cite{Vishik1985}, corresponds to
functions \(F\in H_0(\Omega_\sigma)\) for which
\(r^{h+1}\|F\|_{\Omega,r}\to0\) as \(r\to0\).
\end{remark}

\subsection{Spectral distributions from resolvents}\label{sec:spectral-distribution}

We are now in a position to associate a $p$-adic spectral distribution to a suitable operator via the inverse Stieltjes theorem. This is the key construction that allows us to define spectral zeta functions for continuous spectra. We first introduce the class of operators to which our method applies, then prove that their resolvents give rise to unique generalized distributions.

\begin{definition}[{Admissible operator; cf. \cite[pp.~138--142, Theorem]{Koblitz2} and \cite{Vishik1985}}]
\label{def:admissible-operator}
Let $T$ be a bounded linear operator on a $\Cp$-Banach space. We say that $T$ is \emph{admissible} if there exists a compact set $\sigma\subset\Cp$ such that:
\begin{enumerate}[(1)]
\item $\sigma(T)\subset\sigma$, where $\sigma(T)$ denotes the spectrum of $T$;
\item for every \(z\in\Omega_\sigma\), the resolvent \(R(z)=(zI-T)^{-1}\) exists as a bounded operator;
\item there exists a continuous linear functional $\ell$ on the Banach space such that the scalar function
\[
F_T(z):=\ell(R(z))
\]
belongs to $H_0(\Omega_{\sigma})$.
\end{enumerate}
\end{definition}

\begin{remark}
The notion of an admissible operator is not explicitly defined in \cite{Koblitz2}.
It abstracts the conditions under which the inverse Stieltjes transform of \cite[pp.~138--142, Theorem]{Koblitz2} and the Vishik spectral theorem of \cite{Vishik1985} can be applied.
In particular, condition (3) ensures that the scalar resolvent $F_T$ lies in $H_0(\Omega_{\sigma})$, so that the inverse Stieltjes theorem produces a unique generalized distribution $\mu_T\in\cL^*(\sigma)$.
\end{remark}

\begin{definition}[{$p$-adic spectral distribution, cf. \cite[pp.~138--142, Theorem]{Koblitz2}}]
\label{def:spectral-distribution}
For an admissible operator $T$ with associated function $F_T\in H_0(\Omega_{\sigma})$, let $\mu_T\in\cL^*(\sigma)$ be the unique generalized distribution satisfying
\[
S_{\mu_T}(z)=F_T(z)
\]
for all $z\in\Omega_{\sigma}$. We call $\mu_T$ the \emph{$p$-adic spectral distribution} of $T$ (relative to the functional $\ell$).
\end{definition}

\begin{remark}
The existence and uniqueness of $\mu_T$ follow from the inverse Stieltjes theorem
\cite[pp.~138--142, Theorem]{Koblitz2}, which states that
$S:\cL^*(\sigma)\to H_0(\Omega_{\sigma})$ is a bijection.
The term ``$p$-adic spectral distribution'' is not used in \cite{Koblitz2};
it is introduced here to emphasize the spectral interpretation.
\end{remark}
The following proposition gives a concrete way to compute the spectral distribution in terms of the resolvent expansion at infinity.

\begin{proposition}[{Expansion at infinity; cf. \cite[pp.~138--142, Theorem]{Koblitz2}}]
\label{prop:expansion}
Let $T$ be an admissible operator with $\sigma\subset\{z\in\Cp:|z|_p\le R\}$ for some $R>0$. Then for $|z|_p>R$, we have
\[
F_T(z)=\sum_{m=0}^{\infty}\frac{\mu_T(x^m)}{z^{m+1}},
\]
where $\mu_T(x^m)$ denotes the action of the generalized distribution $\mu_T$ on the monomial $x^m$.
\end{proposition}

\begin{proof}
By Definition~\ref{def:spectral-distribution}, \(F_T=S_{\mu_T}\), where
\[
S_{\mu_T}(z)=\left(\mu_T(x),\frac{1}{z-x}\right),\quad z\notin\sigma.
\]
Let \(|z|_p>R\). Since \(\sigma\subset\{z\in\Cp:|z|_p\le R\}\), for every \(x\in\sigma\) we have \(|x/z|_p<1\), so the geometric series
\[
\frac{1}{z-x}
=
\frac{1}{z}\cdot\frac{1}{1-x/z}
=
\sum_{m=0}^{\infty}\frac{x^m}{z^{m+1}}
\]
converges uniformly on \(\sigma\). Because \(\mu_T\in\cL^*(\sigma)\) is continuous, it commutes with this uniformly convergent series, and therefore
\[
F_T(z)
=
S_{\mu_T}(z)
=
\left(\mu_T(x),\frac{1}{z-x}\right)
=
\sum_{m=0}^{\infty}\frac{\mu_T(x^m)}{z^{m+1}}.
\]
This proves the stated expansion. See also \cite[pp.~140--141]{Koblitz2} for the same computation in the proof of Vishik's inverse Stieltjes theorem.
\end{proof}

\section{Bosonic and fermionic $p$-adic spectral zeta functions}\label{sec:zeta}

Having constructed the \(p\)-adic spectral distribution
\(\mu_T\in\cL^*(\sigma)\) from the resolvent in
Section~\ref{sec:spectral-distribution}, we now define the associated
spectral zeta functions. Two parallel constructions are considered,
following the two measures introduced in Section~\ref{sec:measures}:
the bosonic zeta function, modeled on the unbounded Haar distribution,
and the fermionic zeta function, modeled on the bounded
\(\mu_{-1}\)-measure. This distinction is reflected in the analytic
properties of the two functions: the bosonic case requires estimates
in differentiable function spaces (in particular, \(C^2\)-estimates in
the compatibility proof), whereas the fermionic case only requires
continuity.

For each case we establish the analyticity of the zeta function in
\((s,\lambda)\), derive special value formulas at non-positive
integers, and prove a compatibility theorem showing that the zeta
function of \cite{HuKim} is recovered as the zeta function associated
with a spectral distribution constructed by the inverse Stieltjes
transform.

\subsection{Bosonic zeta function}

We begin with the bosonic case, which is an extension of our previous zeta function with the Haar distribution replaced by a general spectral distribution $\mu$.

\begin{definition}[Support of a generalized distribution]
\label{def:support}
Let \(\sigma\subset\Cp\) be compact and let
\(\mu\in\cL^*(\sigma)\). We say that \(\mu\) vanishes on a
relatively open subset \(U\subset\sigma\) if
\[
\mu(f)=0
\quad\text{for every }f\in\cL(\sigma)
\text{ with }\supp(f)\subset U,
\]
where
\[
\supp(f)
=
\overline{\{x\in\sigma:f(x)\neq0\}}^{\,\sigma}
\]
denotes the closure in \(\sigma\) of the nonzero set of \(f\).
The \emph{support} of \(\mu\) is defined by
\[
\supp(\mu)
=
\sigma\setminus
\bigcup_{\substack{U\subset\sigma\ \mathrm{relatively\ open}\\
\mu\ \mathrm{vanishes\ on}\ U}}
U.
\]
\end{definition}

\begin{definition}[{Bosonic $p$-adic spectral zeta function, see \cite[Definition I.1]{HuKim}}]
\label{def:bosonic-zeta}
Let $\mu\in\cL^*(\sigma)$ be a $p$-adic spectral distribution. For $s\in\Zp\setminus\{1\}$ and $\lambda\in\Cp$ such that $-\lambda\notin\supp(\mu)$, define the \emph{bosonic $p$-adic spectral zeta function} by
\[
\zeta_p^\mu(s,\lambda)=\frac{1}{s-1}\int_\sigma \langle \lambda+x\rangle^{1-s}\,\d\mu(x).
\]
\end{definition}

Our goal in this subsection is to prove the analyticity theorem
(Theorem~\ref{thm:welldef}), which extends
\cite[Theorem III.2]{HuKim}. The proof follows the same overall
strategy as that of \cite[Theorem III.2]{HuKim}. However, in the
present setting \(\sigma\) is an arbitrary compact subset of \(\Cp\)
rather than \(\Zp\), and \(\mu\) is a generalized distribution rather
than a concrete measure. These two generalizations introduce several
technical complications. To handle them in a clean and modular way, we
first construct the open set \(\Omega\) on which \(\lambda+x\) stays
away from zero, and then isolate the relevant properties of the
auxiliary function \(g_{\lambda,s}(x)=\langle\lambda+x\rangle^{1-s}\)
in Lemma~\ref{lem:g-lambda-s}.

We begin with the construction of \(\Omega\). Since \(\sigma\subset\Cp\)
is compact and the map \(x\mapsto|x|_p\) is continuous, the maximum
\[
\max_{x\in\sigma}|x|_p
\]
exists and is finite. Choose \(\lambda_0\in\Cp\) with
\[
|\lambda_0|_p>\max_{x\in\sigma}|x|_p.
\]
Then
\[
\delta:=\operatorname{dist}(-\lambda_0,\sigma)
=\min_{x\in\sigma}|\lambda_0+x|_p>0.
\]
Define
\begin{equation}\label{omega}
\Omega=\{\lambda\in\Cp:|\lambda-\lambda_0|_p<\delta/2\}.
\end{equation}
For every \(\lambda\in\Omega\) and every \(x\in\sigma\), we have
\[
|\lambda-\lambda_0|_p<\delta/2<\delta\le|\lambda_0+x|_p.
\]
Applying the ultrametric inequality with \(a=\lambda_0+x\) and
\(b=\lambda-\lambda_0\), we obtain
\[
|\lambda+x|_p=|(\lambda_0+x)+(\lambda-\lambda_0)|_p
=|\lambda_0+x|_p\ge\delta>0.
\]
Thus \(\lambda+x\neq0\) for all \(\lambda\in\Omega\) and
\(x\in\sigma\). In particular, \(\langle\lambda+x\rangle\) is
well-defined on \(\sigma\) and belongs to \(1+\mathfrak m_{\Cp}\).
Since \(\operatorname{supp}(\mu)\subset\sigma\), the same holds for
every \(x\in\operatorname{supp}(\mu)\).

Having constructed \(\Omega\), we now isolate the technical properties
of the \(p\)-adic power--exponential function
\(g_{\lambda,s}(x)=\langle\lambda+x\rangle^{1-s}\) that will be used
in the proof of the main analyticity theorem. This is the content of
the following lemma.

Recall from Definition~\ref{def:Br} that for \(r>0\),
\[
\mathbb D_\sigma(r)=\bigcup_{a\in\sigma}D_a(r),
\quad
D_a(r)=\{x\in\mathbb C_p:|x-a|_p<r\},
\]
and \(B_r(\sigma)\) is the space of functions on \(\mathbb D_\sigma(r)\) that are given by a convergent power series on each disc \(D_a(r)\), equipped with the norm
\[
\|f\|_r=\sup_{x\in\mathbb D_\sigma(r)}|f(x)|_p.
\]
Recall also that \(\cL(\sigma)=\bigcup_{r>0}B_r(\sigma)\) is the space of locally analytic functions on \(\sigma\), and that \(\cL^*(\sigma)\) denotes its continuous dual.

\begin{lemma}[{Analyticity of $\langle\lambda+x\rangle^{1-s}$}] \label{lem:g-lambda-s}
Let \(\sigma\subset\Cp\) be compact, and let \(\Omega\) be as in \eqref{omega}. For every \(\lambda\in\Omega\) and \(s\in\Zp\), define
\[
g_{\lambda,s}(x):=\langle\lambda+x\rangle^{1-s}.
\]
Then the following hold.

\begin{enumerate}
\item[(i)] For each \(\lambda\in\Omega\), there exists \(r>0\) such that \(g_{\lambda,s}\in B_r(\sigma)\) for all \(s\in\Zp\), and
\[
M(\lambda):=\sup_{x\in\mathbb D_\sigma(r)}\bigl|\log_p\langle\lambda+x\rangle\bigr|_p<\infty.
\]

\item[(ii)] Fix \(s_0\in\Zp\setminus\{1\}\), and let \(J\subset\Zp\setminus\{1\}\) be a compact neighbourhood of \(s_0\). There exists \(r>0\), independent of \(s\in J\), such that for \(h=s-s_0\),
\[
g_{\lambda,s}(x)
=
\sum_{n=0}^{\infty}h^n\psi_n(x),
\quad
\psi_n(x)
=
\frac{(-1)^n}{n!}
\langle\lambda+x\rangle^{1-s_0}
\bigl(\log_p\langle\lambda+x\rangle\bigr)^n,
\]
with \(\psi_n\in B_r(\sigma)\) for all \(n\ge0\), and the series converges absolutely in \(B_r(\sigma)\) whenever
\[
|h|_p<\frac{p^{-1/(p-1)}}{M(\lambda)}.
\]

\item[(iii)] For every compact \(K\subset\Omega\), there exist \(r>0\) and \(\rho<1\), independent of \(\lambda\in K\), such that \(g_{\lambda,s}\in B_r(\sigma)\) for all \(\lambda\in K\) and
\[
|\langle\lambda+x\rangle-1|_p\le\rho
\quad(\lambda\in K,\ x\in\mathbb D_\sigma(r)).
\]
\end{enumerate}
\end{lemma}

\begin{proof}
(i) We first prove the local boundedness. By construction of
\(\Omega\), \(\lambda+x\neq0\) for all \(x\in\sigma\), so
\(|\langle\lambda+x\rangle-1|_p<1\) for every \(x\in\sigma\). Fix
\(a\in\sigma\). Since \(x\mapsto\langle\lambda+x\rangle\) is
continuous at \(a\), there exists an open neighbourhood \(U_a\) of
\(a\) in \(\sigma\) and a constant \(\rho_a<1\) such that
\[
|\langle\lambda+x\rangle-1|_p\le\rho_a<1
\quad(x\in U_a).
\]
On \(U_a\), the \(\log_p\)-series
\[
\log_p\langle\lambda+x\rangle
=
\sum_{n=1}^{\infty}
\frac{(-1)^{n+1}}{n}\bigl(\langle\lambda+x\rangle-1\bigr)^n
\]
converges uniformly, since its \(n\)-th term is bounded by
\(\rho_a^n\,n\), which is summable because \(\rho_a<1\). Hence there
exists a constant \(C_a<\infty\) such that
\[
|\log_p\langle\lambda+x\rangle|_p\le C_a
\quad(x\in U_a).
\]

Since \(\sigma\) is compact and \(\{U_a\}_{a\in\sigma}\) is an open
cover of \(\sigma\), there exist finitely many points
\(a_1,\dots,a_m\in\sigma\) such that
\[
\sigma\subset U_{a_1}\cup\cdots\cup U_{a_m}.
\]
Set
\[
\rho:=\max_{1\le i\le m}\rho_{a_i}<1,
\quad
C:=\max_{1\le i\le m}C_{a_i}<\infty.
\]
Choose \(r>0\) small enough so that for every \(a\in\sigma\), the
disc \(D_a(r)\) is contained in one of the \(U_{a_i}\). Then for every
\(x\in\mathbb D_\sigma(r)\), we have
\[
|\langle\lambda+x\rangle-1|_p\le\rho<1,
\quad
|\log_p\langle\lambda+x\rangle|_p\le C.
\]
Therefore
\[
M(\lambda)=\sup_{x\in\mathbb D_\sigma(r)}
\bigl|\log_p\langle\lambda+x\rangle\bigr|_p\le C<\infty.
\]
This proves \(M(\lambda)<\infty\).

Finally, since \(\langle\lambda+x\rangle^{1-s}\) and
\(\log_p\langle\lambda+x\rangle\) both belong to \(B_r(\sigma)\), and
\(B_r(\sigma)\) is closed under multiplication, it follows that
\(g_{\lambda,s}\in B_r(\sigma)\) for all \(s\in\Zp\).

(ii) We expand \(g_{\lambda,s}\) in powers of \(h=s-s_0\).
Since \(\lambda+x\neq0\) on \(\sigma\), the function
\(\log_p\langle\lambda+x\rangle\) is well-defined and, by part~(i),
belongs to \(B_r(\sigma)\). Writing
\[
g_{\lambda,s}(x)
=
\langle\lambda+x\rangle^{1-s_0}
\exp_p\!\bigl(-h\log_p\langle\lambda+x\rangle\bigr),
\]
we expand the exponential as
\[
\exp_p\!\bigl(-h\log_p\langle\lambda+x\rangle\bigr)
=
\sum_{n=0}^{\infty}
\frac{(-1)^n}{n!}
h^n\bigl(\log_p\langle\lambda+x\rangle\bigr)^n.
\]
Substituting this into the expression for \(g_{\lambda,s}\) yields
\[
g_{\lambda,s}(x)
=
\sum_{n=0}^{\infty}h^n\psi_n(x),
\quad
\psi_n(x)
=
\frac{(-1)^n}{n!}
\langle\lambda+x\rangle^{1-s_0}
\bigl(\log_p\langle\lambda+x\rangle\bigr)^n.
\]
The radius \(r\) is independent of \(n\). Indeed, both factors
\(\langle\lambda+x\rangle^{1-s_0}\) and
\(\log_p\langle\lambda+x\rangle\) lie in \(B_r(\sigma)\) by part~(i),
and \(B_r(\sigma)\) is closed under multiplication; hence
\(\psi_n\in B_r(\sigma)\) for every \(n\ge0\), with the same \(r\)
that was fixed in part~(i).

To prove absolute convergence in \(B_r(\sigma)\), we estimate
\(\|\psi_n\|_r\). By the definition of the norm on \(B_r(\sigma)\),
\[
\|\psi_n\|_r
\le
\frac{1}{|n!|_p}
\bigl\|\langle\lambda+x\rangle^{1-s_0}\bigr\|_r
\sup_{x\in\mathbb D_\sigma(r)}
\bigl|\log_p\langle\lambda+x\rangle\bigr|_p^n.
\]
By the definition of \(M(\lambda)\) in part~(i),
\[
\sup_{x\in\mathbb D_\sigma(r)}
\bigl|\log_p\langle\lambda+x\rangle\bigr|_p
=
M(\lambda),
\]
so
\[
\|\psi_n\|_r
\le
\frac{1}{|n!|_p}
\bigl\|\langle\lambda+x\rangle^{1-s_0}\bigr\|_r
M(\lambda)^n.
\]
Using the standard lower bound \(1/|n!|_p\le p^{\,n/(p-1)}\)
(see \cite[p.~21, Lemma~3]{Iw}), we obtain
\[
\|\psi_n\|_r
\le
C_0(\lambda,s_0)\bigl(M(\lambda)p^{1/(p-1)}\bigr)^n,
\]
where
\[
C_0(\lambda,s_0):=\bigl\|\langle\lambda+x\rangle^{1-s_0}\bigr\|_r<\infty
\]
depends only on \(\lambda\) and \(s_0\), not on \(n\). Therefore the
series \(\sum_{n=0}^{\infty}h^n\psi_n\) converges absolutely in
\(B_r(\sigma)\) whenever
\[
|h|_p\,M(\lambda)\,p^{1/(p-1)}<1,
\quad\text{i.e.,}\quad
|h|_p<\frac{p^{-1/(p-1)}}{M(\lambda)}.
\]

(iii) Fix a compact \(K\subset\Omega\). By the argument of part~(i)
(applied with \(\lambda\) replaced by an arbitrary point of \(K\), and
using the continuity of \(\lambda\mapsto\langle\lambda+x\rangle\)), for
each \(\lambda\in K\) there exist an open neighbourhood
\(U_\lambda\subset\Omega\) of \(\lambda\), a radius \(r_\lambda>0\),
and a constant \(\rho_\lambda<1\) such that for all
\(\lambda'\in U_\lambda\), both \(\langle\lambda'+x\rangle^{1-s}\)
(for all \(s\in\Zp\)) and \(\log_p\langle\lambda'+x\rangle\) belong to
\(B_{r_\lambda}(\sigma)\), and
\[
|\langle\lambda'+x\rangle-1|_p\le\rho_\lambda
\quad(x\in\mathbb D_\sigma(r_\lambda)).
\]
Since \(K\) is compact and \(\{U_\lambda\}_{\lambda\in K}\) is an open
cover of \(K\), there exist finitely many points
\(\lambda_1,\dots,\lambda_m\in K\) such that
\(K\subset U_{\lambda_1}\cup\cdots\cup U_{\lambda_m}\). Set
\[
r:=\min_{1\le i\le m}r_{\lambda_i}>0,
\quad
\rho:=\max_{1\le i\le m}\rho_{\lambda_i}<1.
\]
Since \(r\le r_{\lambda_i}\) for every \(i\), we have
\(\mathbb D_\sigma(r)\subset\mathbb D_\sigma(r_{\lambda_i})\); as
restriction to a smaller disc preserves membership in
\(B_r(\sigma)\), it follows that \(\langle\lambda_i+x\rangle^{1-s}\)
(for all \(s\in\Zp\)) and \(\log_p\langle\lambda_i+x\rangle\) belong
to \(B_r(\sigma)\). Moreover, for \(x\in\mathbb D_\sigma(r)\subset
\mathbb D_\sigma(r_{\lambda_i})\),
\[
|\langle\lambda_i+x\rangle-1|_p\le\rho_{\lambda_i}\le\rho<1.
\]
Now let \(\lambda\in K\) be arbitrary, and choose \(i\) with
\(\lambda\in U_{\lambda_i}\). Since \(r\le r_{\lambda_i}\) and
\(\rho_{\lambda_i}\le\rho\), the preceding two displays give
\[
g_{\lambda,s}\in B_r(\sigma),
\quad
|\langle\lambda+x\rangle-1|_p\le\rho
\quad(x\in\mathbb D_\sigma(r)).
\]
This proves part~(iii).
\end{proof}

With the analytic properties for \(g_{\lambda,s}\) established in
Lemma~\ref{lem:g-lambda-s}, we are now ready to prove the main
analyticity theorem. It states that the bosonic spectral zeta function
is well-defined and possesses the expected smoothness in \(s\) and
local analyticity in \(\lambda\).
\begin{theorem}[{Analyticity; cf. \cite[Theorem III.2]{HuKim}}]
\label{thm:welldef}
There exists a nonempty open set $\Omega\subset\Cp$ such that for every
$\lambda\in\Omega$ and every $s\in\Zp\setminus\{1\}$, the action in
Definition \ref{def:bosonic-zeta} is well-defined. Moreover,
$\zeta_p^\mu(s,\lambda)$ is $C^\infty$ in $s$ on $\Zp\setminus\{1\}$
and locally analytic in $\lambda$ on $\Omega$.
\end{theorem}

\begin{remark}[Quantitative refinement in the discrete case]
\label{rem:quantitative}
When $\supp(\mu)$ is contained in a finite extension $K$ of $\Q_p$
with ramification index $e<p-1$, and $(\pi)$ denotes the maximal ideal
of the ring of integers of $K$, the argument below can be refined to
yield the sharper statement that $\zeta_p^\mu(s,\lambda)$ is analytic
for
\[
|s|_p<|\pi|_p^{-1}p^{-1/(p-1)}
\]
except for a simple pole at $s=1$, while the fermionic counterpart
$\zeta_{p,E}^\nu(s,\lambda)$ is analytic on the same disc. For general
$\mu\in\cL^*(\sigma)$, such a quantitative bound is not available,
since $\supp(\mu)$ may contain transcendental points that do not lie
in any finite extension of $\Q_p$.
\end{remark}

\begin{proof}
Let \(\Omega\) be the open set defined in \eqref{omega}.
By Lemma~\ref{lem:g-lambda-s}(i), for each \(\lambda\in\Omega\)
there exists \(r>0\) such that
\[
g_{\lambda,s}(x)=\langle\lambda+x\rangle^{1-s}
\in B_r(\sigma)\subset\cL(\sigma)
\]
for every \(s\in\Zp\). Hence \(\mu(g_{\lambda,s})\), and therefore
\(\zeta_p^\mu(s,\lambda)\) for \(s\neq1\), is well-defined.

\medskip
\noindent\textbf{\(C^\infty\)-smoothness in \(s\).}
Fix \(\lambda\in\Omega\) and \(s_0\in\Zp\setminus\{1\}\).
By the local expansion in Lemma~\ref{lem:g-lambda-s}(ii),
for \(h=s-s_0\) sufficiently small we have
\[
g_{\lambda,s}
=
\sum_{n=0}^{\infty}h^n\psi_n,
\]
with convergence in a fixed space \(B_r(\sigma)\).
Since \(\mu\) is continuous on \(B_r(\sigma)\), applying it
term by term gives
\[
\mu(g_{\lambda,s})
=
\sum_{n=0}^{\infty}h^n\mu(\psi_n).
\]
Thus \(s\mapsto\mu(g_{\lambda,s})\) is locally analytic near
\(s_0\), and in particular is \(C^\infty\). Differentiation
of this convergent power series also yields
\[
\frac{\d^k}{\d s^k}\mu(g_{\lambda,s})
=
\mu\!\left(
\frac{\partial^k}{\partial s^k}g_{\lambda,s}
\right),
\qquad k\ge0.
\]
Since \((s-1)^{-1}\) is analytic away from \(s=1\),
\(\zeta_p^\mu(s,\lambda)\) is \(C^\infty\) in \(s\) on
\(\Zp\setminus\{1\}\).

\medskip
\noindent\textbf{Local analyticity in \(\lambda\).}
Fix \(s\in\Zp\setminus\{1\}\) and \(\lambda_*\in\Omega\).
By compactness of \(\sigma\), continuity of the projection,
and local constancy of \(\omega_v\), we may choose \(r>0\)
and a closed disc
\[
U=\{\lambda\in\Cp:|\lambda-\lambda_*|_p\le d\}
\subset\Omega,
\qquad d\in|\Cp^\times|_p,
\]
small enough that, for all \(\lambda\in U\) and
\(x\in\mathbb D_\sigma(r)\),
\[
\omega_v(\lambda+x)=\omega_v(\lambda_*+x)
\]
and, for some \(\rho<1\),
\begin{equation}\label{rho}
g_{\lambda,s}\in B_r(\sigma),
\qquad
|\langle\lambda+x\rangle-1|_p\le\rho.
\end{equation}
These choices follow from the same local estimates and
finite-cover argument used in Lemma~\ref{lem:g-lambda-s}.

Define
\[
P_N(\lambda,x)
=
\sum_{n=0}^{N}\binom{1-s}{n}
\bigl(\langle\lambda+x\rangle-1\bigr)^n.
\]
On \(U\), the identity
\[
\langle\lambda+x\rangle
=
\frac{\lambda+x}{\omega_v(\lambda_*+x)}
\]
shows that \(P_N(\lambda,\cdot)\) is a polynomial in
\(\lambda\) with coefficients in \(B_r(\sigma)\).
Consequently, \(\mu(P_N(\lambda,\cdot))\) is a scalar
polynomial in \(\lambda\).

Since \(\binom{1-s}{n}\in\Zp\), \eqref{rho} gives
\[
\sup_{\lambda\in U}
\|P_N(\lambda,\cdot)-g_{\lambda,s}\|_r
\le\rho^{N+1}\longrightarrow0.
\]
By continuity of \(\mu\),
\[
\sup_{\lambda\in U}
\left|
\mu(P_N(\lambda,\cdot))-\mu(g_{\lambda,s})
\right|_p
\le
\|\mu\|_r\,\rho^{N+1}\longrightarrow0.
\]
Thus \(\lambda\mapsto\mu(g_{\lambda,s})\) is a uniform
limit of polynomials on the closed disc \(U\), and hence
is analytic there; see
\cite[Theorem~42.2(ii), Theorem~42.3(ii)]{SC}.
Since \(\lambda_*\) was arbitrary,
\(\zeta_p^\mu(s,\lambda)=(s-1)^{-1}\mu(g_{\lambda,s})\)
is locally analytic in \(\lambda\) on \(\Omega\).
\end{proof}

To express the special values of $\zeta_p^\mu(s,\lambda)$, we need to extend the classical Bernoulli polynomial $B_m(\lambda)$ and Euler polynomial
$E_m(\lambda)$ to the present setting. We begin by recalling their classical generating functions.
\begin{equation}\label{eq:classical-BE}
\frac{t e^{\lambda t}}{e^{t}-1}
=
\sum_{m=0}^{\infty}B_m(\lambda)\frac{t^m}{m!},
\quad
\frac{2 e^{\lambda t}}{e^{t}+1}
=
\sum_{m=0}^{\infty}E_m(\lambda)\frac{t^m}{m!},
\end{equation}
see \cite[Chapter 11]{Cohen}. In the \(p\)-adic setting these polynomials admit the integral representations
\begin{equation}\label{eq:p-adic-BE}
B_m(\lambda)=\int_{\Zp}(\lambda+a)^m\,\d a,
\quad
E_m(\lambda)=\int_{\Zp}(\lambda+a)^m\,\d\mu_{-1}(a),
\end{equation}
where the integral with respect to \(da\) is the Volkenborn integral with respect to the Haar distribution, and \(\mu_{-1}\) is the measure defined in Definition \ref{def:mu-1}. The first identity is standard (see \cite[Lemma 11.1.7]{Cohen}), while the second follows from the properties of \(\mu_{-1}\) (see \cite[(2.6)]{HK2012}).

In our previous work \cite{HuKim}, for a locally analytic and bounded interpolation function \(f:\Zp\to\Cp\), we introduced the associated Bernoulli and Euler polynomials
\begin{equation}\label{eq:Bf-Ef}
B_m^f(\lambda)=\int_{\Zp}(\lambda+f(a))^m\,\d a,
\quad
E_m^f(\lambda)=\int_{\Zp}(\lambda+f(a))^m\,\d\mu_{-1}(a),
\end{equation}
and proved the special value formulas
\begin{equation}\label{eq:special-f}
\zeta_p^f(1-m,\lambda)
=
-\frac{1}{\omega_v^m(\lambda)}\frac{B_m^f(\lambda)}{m},
\quad
\zeta_{p,E}^f(1-m,\lambda)
=
\frac{1}{\omega_v^m(\lambda)}E_m^f(\lambda),
\end{equation}
for \(m\ge1\); see \cite[(45), (46), (53), (54)]{HuKim}. The definitions \eqref{eq:Bf-Ef} reduce to \eqref{eq:p-adic-BE} when \(f(a)=a\), i.e.\ for the integer spectrum.

We now extend \eqref{eq:Bf-Ef} to arbitrary generalized spectral distributions. Let \(\mu\in\cL^*(\sigma)\). In analogy with \eqref{eq:p-adic-BE} and \eqref{eq:Bf-Ef}, we define the \emph{Bernoulli functionals} associated with \(\mu\) by
\begin{equation}\label{eq:Bmu}
B_m^\mu(\lambda)=\int_\sigma(\lambda+x)^m\,\d\mu(x),
\quad m\ge0.
\end{equation}
In particular,
\[
B_m^\mu(0)=\int_\sigma x^m\,\d\mu(x).
\]
If \(\mu=\mu_f\) is the pushforward of the Haar distribution under a locally analytic interpolation function \(f\), then \(B_m^{\mu_f}(\lambda)=B_m^f(\lambda)\). Thus \eqref{eq:Bmu} is a genuine generalization of \eqref{eq:Bf-Ef} and \eqref{eq:p-adic-BE}.

\begin{definition}[{Bernoulli functionals, see \eqref{eq:Bmu}}]
\label{def:Bernoulli}
For $\mu\in\cL^*(\sigma)$ and $m\ge0$, the \emph{Bernoulli functionals} are defined by
\[
B_m^\mu(\lambda)=\int_\sigma (\lambda+x)^m\,\d\mu(x).
\]
In particular, $B_m^\mu(0)=\int_\sigma x^m\,\d\mu(x)$.
\end{definition}

\begin{theorem}[{Special value formula, bosonic; cf. \cite[(53)]{HuKim}}]
\label{thm:special-bosonic}
For $m\ge1$ and $|\lambda|_p$ sufficiently large,
\[
\zeta_p^\mu(1-m,\lambda)
=
-\frac{1}{\omega_v^m(\lambda)}
\frac{B_m^\mu(\lambda)}{m}.
\]
\end{theorem}

\begin{proof}
Assume that
\[
|\lambda|_p>\max_{x\in\supp(\mu)}|x|_p.
\]
Then, for every \(x\in\supp(\mu)\),
\[
|(\lambda+x)-\lambda|_p=|x|_p<|\lambda|_p.
\]
By the local constancy property established in
Subsection~\ref{sec:p-adic-numbers}, the valuation and
Teichm\"uller factors of \(\lambda+x\) coincide with those
of \(\lambda\). Hence
\[
\omega_v(\lambda+x)=\omega_v(\lambda).
\]
Using the definition
\(\langle a\rangle=a/\omega_v(a)\), we obtain
\[
\langle\lambda+x\rangle
=
\frac{\lambda+x}{\omega_v(\lambda)}
=
\langle\lambda\rangle
\left(1+\frac{x}{\lambda}\right).
\]
Since \(|x/\lambda|_p<1\), the factor \(1+x/\lambda\)
belongs to \(1+\mathfrak m_{\Cp}\). Thus, for \(s\in\Zp\),
the multiplicativity of \(p\)-adic powers on
\(1+\mathfrak m_{\Cp}\), together with their binomial
series definition, gives
\[
\begin{aligned}
\langle\lambda+x\rangle^{1-s}
&=
\langle\lambda\rangle^{1-s}
\left(1+\frac{x}{\lambda}\right)^{1-s}\\
&=
\langle\lambda\rangle^{1-s}
\sum_{n=0}^{\infty}
\binom{1-s}{n}\frac{x^n}{\lambda^n}.
\end{aligned}
\]
The series converges uniformly for \(x\in\supp(\mu)\), since
\[
\sup_{x\in\supp(\mu)}
\left|\frac{x}{\lambda}\right|_p<1
\quad\text{and}\quad
\binom{1-s}{n}\in\Zp.
\]

Then integrating term by term against \(\mu\) and setting \(s=1-m\), we obtain
\[
\zeta_p^\mu(1-m,\lambda)
=
-\frac{1}{m}\langle\lambda\rangle^{m}
\sum_{n=0}^{\infty}\binom{m}{n}\frac{\mu(x^n)}{\lambda^n}.
\]
Since \(\langle\lambda\rangle=\omega_v^{-1}(\lambda)\lambda\), by (\ref{eq:Bmu}) the sum equals \(\omega_v^{-m}(\lambda)B_m^\mu(\lambda)\), and the result follows.
\end{proof}

\begin{remark}
When the spectrum is discrete and admits a locally analytic interpolation function \(f\), the bosonic spectral distribution \(\mu\) is the pushforward of the Haar distribution under \(f\). In that case \(B_m^\mu(\lambda)=B_m^f(\lambda)\), and Theorem~\ref{thm:special-bosonic} reduces to \cite[(53)]{HuKim}. Thus the present construction extends our previous framework.
\end{remark}

\subsection{Fermionic resolvent and zeta function}

The fermionic version is obtained by using a bounded spectral distribution, associated with the \(\mu_{-1}\)-measure. Recall from Definition~\ref{def:mu-1} that \(\mu_{-1}\) is bounded:
\[
|\mu_{-1}(a+p^N\Zp)|_p=|(-1)^a|_p=1
\]
for all \(a\) and \(N\). Consequently, the corresponding integral is defined for every continuous function, without any differentiability requirement. This makes the fermionic case simpler than the bosonic one, where the Haar distribution is unbounded and the integral is defined only for \(C^1\) functions. We begin by defining the fermionic resolvent and the associated spectral distribution (see Definition~\ref{def:fermionic-resolvent}), and then use them to define the fermionic zeta function (see Definition~\ref{def:fermionic-zeta}).

\begin{definition}[{Fermionic resolvent}]
\label{def:fermionic-resolvent}
Let \(f:\Zp\to\sigma\) be continuous. The \emph{fermionic resolvent} associated with \(f\) is the function
\[
F_f^E(z)=\int_{\Zp}\frac{\d\mu_{-1}(a)}{z-f(a)},\quad z\in\Omega_{\sigma}:=\Cp\setminus\sigma.
\]
Its Riemann sums are
\[
F_{f,N}^E(z)=\sum_{a=0}^{p^N-1}\frac{(-1)^a}{z-f(a)},\quad N\ge0.
\]
\end{definition}

\begin{proposition}\label{prop:fermionic-resolvent-well-defined}
The fermionic resolvent \(F_f^E\) is well-defined, and its Riemann sums \(F_{f,N}^E\) converge uniformly on each
\[
\Omega_\sigma(r)=\{z\in\Cp:\dist(z,\sigma)\ge r\},\quad r>0.
\]
Moreover, \(F_f^E\in H_0(\Omega_{\sigma})\).
\end{proposition}

\begin{proof}
Since \(f\) is continuous and \(\mu_{-1}\) is bounded, the integral defining \(F_f^E(z)\) converges for every \(z\in\Omega_{\sigma}\).

To prove uniform convergence of the Riemann sums, fix \(r>0\) and let \(z\in\Omega_\sigma(r)\). For every \(a\in\Zp\), we have \(|z-f(a)|_p\ge r\). Hence the function
\[
g_z(a)=\frac{1}{z-f(a)}
\]
is continuous on \(\Zp\).

We estimate the error between the integral and its Riemann sum. For each \(N\ge0\),
decomposing \(\Zp\) into the \(p^N\) disjoint clopen residue classes modulo \(p^N\), and using the finite additivity of the distribution \(\mu_{-1}\), we obtain
\[
\int_{\Zp} g_z(a)\,\d\mu_{-1}(a)
=
\sum_{a=0}^{p^N-1} \int_{a+p^N\Zp} g_z(x)\,\d\mu_{-1}(x).
\]
Since \(\mu_{-1}(a+p^N\Zp)=(-1)^a\), we may write
\[
\int_{\Zp} g_z(a)\,\d\mu_{-1}(a)
=
\sum_{a=0}^{p^N-1} g_z(a)(-1)^a
+
\sum_{a=0}^{p^N-1} \int_{a+p^N\Zp} \bigl(g_z(x)-g_z(a)\bigr)\,\d\mu_{-1}(x).
\]
Therefore by the strong triangle inequality, we have
\[
\left|
\int_{\Zp} g_z(a)\,\d\mu_{-1}(a)
-\sum_{a=0}^{p^N-1} g_z(a)(-1)^a
\right|_p
\le
\max_{0\le a<p^N} \sup_{x\in a+p^N\Zp} |g_z(x)-g_z(a)|_p,
\]
where we used \(|\mu_{-1}(a+p^N\Zp)|_p=1\).

Now \(f\) is continuous on the compact set \(\Zp\), hence uniformly continuous. Let \(\omega_f\) denote its modulus of continuity. For \(x,a\in a+p^N\Zp\), we have \(|x-a|_p\le p^{-N}\), so
\[
|f(x)-f(a)|_p \le \omega_f(p^{-N}).
\]
Since \(|z-f(x)|_p\ge r\) and \(|z-f(a)|_p\ge r\), we obtain
\[
\begin{aligned}
|g_z(x)-g_z(a)|_p
&= \left|\frac{1}{z-f(x)}-\frac{1}{z-f(a)}\right|_p \\
&= \frac{|f(x)-f(a)|_p}{|z-f(x)|_p\,|z-f(a)|_p}
\le \frac{\omega_f(p^{-N})}{r^2}.
\end{aligned}
\]
This bound is independent of \(z\in\Omega_\sigma(r)\). Hence
\[
\left|
\int_{\Zp} g_z(a)\,\d\mu_{-1}(a)
-\sum_{a=0}^{p^N-1} g_z(a)(-1)^a
\right|_p
\le
\frac{\omega_f(p^{-N})}{r^2}
\longrightarrow 0
\quad (N\to\infty),
\]
uniformly in \(z\in\Omega_\sigma(r)\). Thus \(F_f^E\) is a uniform limit of rational functions with poles in \(\sigma\), i.e., \(F_f^E\) is Krasner analytic on \(\Omega_{\sigma}\).

It remains to prove that \(F_f^E\) vanishes at infinity. Since \(f\) is continuous on the compact set \(\Zp\), it is bounded: there exists \(M<\infty\) such that
\[
|f(a)|_p\le M\quad(a\in\Zp).
\]
For \(|z|_p>M\) and every \(a\in\Zp\), we have
\[
|z-f(a)|_p=|z|_p.
\]
Therefore
\[
\begin{aligned}
|F_f^E(z)|_p
=
\left|
\int_{\Zp}\frac{\d\mu_{-1}(a)}{z-f(a)}
\right|_p
&\le
\sup_{a\in\Zp}\frac{1}{|z-f(a)|_p}\int_{\Zp}\d\mu_{-1}(a)\\
&=
\sup_{a\in\Zp}\frac{1}{|z-f(a)|_p}
=
\frac1{|z|_p},
\end{aligned}
\]
by Proposition~\ref{prop:mu-1-total-mass}.
Letting \(|z|_p\to\infty\), we obtain \(F_f^E(z)\to0\). Hence
\(F_f^E\in H_0(\Omega_{\sigma})\) (recall Definition~\ref{def:krasner-analytic}).
\end{proof}

\begin{definition}[{Fermionic spectral distribution}]
\label{def:fermionic-spectral-distribution}
Let \(f:\Zp\to\sigma\) be continuous. By
Proposition~\ref{prop:fermionic-resolvent-well-defined}, the fermionic
resolvent \(F_f^E\) belongs to \(H_0(\Omega_{\sigma})\). Its inverse
Stieltjes transform \(\nu_f^E\in\cL^*(\sigma)\), characterized by
\[
S_{\nu_f^E}=F_f^E,
\]
is called the \emph{fermionic spectral distribution} associated with \(f\).
\end{definition}

\begin{definition}[{Fermionic $p$-adic spectral zeta function, see \cite[Definition I.1]{HuKim}}]
\label{def:fermionic-zeta}
Let \(f:\Zp\to\sigma\) be continuous, and let \(\nu_f^E\in\cL^*(\sigma)\) be the fermionic spectral distribution associated with \(f\) (see Definition~\ref{def:fermionic-spectral-distribution}). Define the \emph{fermionic \(p\)-adic spectral zeta function} by
\[
\zeta_{p,E}^{\nu_f^E}(s,\lambda)
=
\int_\sigma
\langle \lambda+x\rangle^{1-s}
\,\d\nu_f^E(x),
\quad s\in\Zp,\ \lambda\in\Omega.
\]
No factor \((s-1)^{-1}\) is needed because \(\mu_{-1}\) is bounded, and hence so is \(\nu_f^E\).
\end{definition}

The fermionic zeta function enjoys the same analyticity properties as its bosonic counterpart, but the proof is simpler because the underlying measure is bounded. Since \(\nu_f^E\in\cL^*(\sigma)\) is bounded, the integral converges for every continuous function, and no \(C^1\) estimates are needed. We record this as the following theorem.

\begin{theorem}[Analyticity, fermionic]
\label{thm:fermionic-analyticity}
Let \(f:\Zp\to\sigma\) be continuous, and let \(\nu_f^E\in\cL^*(\sigma)\) be the fermionic spectral distribution associated with \(f\) (see Definition~\ref{def:fermionic-spectral-distribution}). There exists a nonempty open set \(\Omega\subset\Cp\) such that for every \(\lambda\in\Omega\) and every \(s\in\Zp\), the integral defining \(\zeta_{p,E}^{\nu_f^E}(s,\lambda)\) (see Definition~\ref{def:fermionic-zeta}) is well-defined. Moreover, \(\zeta_{p,E}^{\nu_f^E}(s,\lambda)\) is \(C^\infty\) in \(s\) on \(\Zp\) and locally analytic in \(\lambda\) on \(\Omega\).
\end{theorem}

\begin{proof}
Since \(\nu_f^E\) is bounded, the integral converges for every continuous function. The function \(g_{\lambda,s}(x)=\langle\lambda+x\rangle^{1-s}\) is continuous on \(\sigma\), and its \(s\)-derivatives are continuous and uniformly bounded on \(\sigma\). Hence differentiation under the integral sign is justified, and the same argument as in Theorem~\ref{thm:welldef} gives the result.
\end{proof}

We now turn to the special values of the fermionic zeta function. We begin by recalling the classical Euler polynomials, which are defined by the generating function
\begin{equation}\label{eq:classical-E}
\frac{2 e^{\lambda t}}{e^{t}+1}
=
\sum_{m=0}^{\infty}E_m(\lambda)\frac{t^m}{m!},
\end{equation}
see \cite[Chapter 11]{Cohen}. In the \(p\)-adic setting they admit the integral representation
\begin{equation}\label{eq:p-adic-E}
E_m(\lambda)=\int_{\Zp}(\lambda+a)^m\,\d\mu_{-1}(a),
\end{equation}
where \(\mu_{-1}\) is the measure defined in Definition \ref{def:mu-1}.

In our previous work \cite{HuKim}, for a locally analytic and bounded interpolation function \(f:\Zp\to\Cp\), we introduced the associated Euler polynomials
\begin{equation}\label{eq:Ef}
E_m^f(\lambda)=\int_{\Zp}(\lambda+f(a))^m\,\d\mu_{-1}(a),
\end{equation}
and proved the special value formula
\begin{equation}\label{eq:special-Ef}
\zeta_{p,E}^f(1-m,\lambda)
=\frac{1}{\omega_v^m(\lambda)}E_m^f(\lambda),
\quad m\ge1,
\end{equation}
see \cite[(46), (54)]{HuKim}. The definition \eqref{eq:Ef} reduces to \eqref{eq:p-adic-E} when \(f(a)=a\).

We now extend \eqref{eq:Ef} to arbitrary fermionic spectral distributions. Let \(\nu\in\cL^*(\sigma)\). In analogy with \eqref{eq:p-adic-E} and \eqref{eq:Ef}, we define the \emph{Euler functionals} associated with \(\nu\) by
\begin{equation}\label{eq:Enu}
E_m^\nu(\lambda)=\int_\sigma(\lambda+x)^m\,\d\nu(x),
\quad m\ge0.
\end{equation}
If \(\nu=\nu_f^E\) is the fermionic spectral distribution of Definition~\ref{def:fermionic-spectral-distribution}, then \(E_m^{\nu_f^E}(\lambda)=E_m^f(\lambda)\).

\begin{theorem}[{Special value formula, fermionic; cf. \cite[(54)]{HuKim}}]
\label{thm:special-fermionic}
For $m\ge0$ and $|\lambda|_p$ sufficiently large,
\[
\zeta_{p,E}^\nu(1-m,\lambda)
=
\frac{1}{\omega_v^m(\lambda)}
E_m^\nu(\lambda).
\]
\end{theorem}

\begin{proof}
For large \(|\lambda|_p\), we have \(|x/\lambda|_p<1\) for all \(x\in\supp(\nu)\), so
\[
\langle\lambda+x\rangle
=
\langle\lambda\rangle\left(1+\frac{x}{\lambda}\right).
\]
Thus
\[
\langle\lambda+x\rangle^{1-s}
=
\langle\lambda\rangle^{1-s}
\sum_{n=0}^{\infty}\binom{1-s}{n}\frac{x^n}{\lambda^n}.
\]
Integrating term by term against \(\nu\) and setting \(s=1-m\), we obtain
\[
\zeta_{p,E}^\nu(1-m,\lambda)
=
\langle\lambda\rangle^{m}
\sum_{n=0}^{\infty}\binom{m}{n}\frac{\nu(x^n)}{\lambda^n}.
\]
Since \(\langle\lambda\rangle=\omega_v^{-1}(\lambda)\lambda\), by (\ref{eq:Enu}) the sum equals \(\omega_v^{-m}(\lambda)E_m^\nu(\lambda)\), and the result follows.
\end{proof}

\begin{remark}
When the spectrum is discrete and admits a locally analytic interpolation function \(f\), the fermionic spectral distribution \(\nu_f^E\) is the pushforward of the \(\mu_{-1}\)-measure under \(f\). In that case \(E_m^{\nu_f^E}(\lambda)=E_m^f(\lambda)\), and Theorem~\ref{thm:special-fermionic} reduces to \cite[(54)]{HuKim}. Thus the present construction extends our previous framework.
\end{remark}

\subsection{Compatibility with the discrete case}
\label{Compatibility}

The main result of this subsection (Theorem \ref{thm:compat}) shows that when the spectrum is discrete and
admits a locally analytic interpolation function, the zeta function of
\cite{HuKim} is recovered as the zeta function associated with a
spectral distribution \(\mu_{\mathcal T}\in\cL^*(\sigma)\) constructed
by the inverse Stieltjes transform. This establishes that the inverse
Stieltjes framework is a genuine extension of our previous theory.
The proof of the bosonic case requires an estimate for the Volkenborn integral, which is most conveniently stated in terms of the spaces \(C^1(\Zp\to\Cp)\) and \(C^2(\Zp\to\Cp)\), together with the notion of the indefinite sum. We therefore begin by recalling these definitions.

\begin{definition}[{\(C^1\)-functions, see \cite[Definition~27.1]{SC}}]
\label{def:C1}
A function \(g:\Zp\to\Cp\) belongs to \(C^1(\Zp\to\Cp)\) if its first difference quotient
\[
\Phi_1 g(a,b)=\frac{g(a)-g(b)}{a-b},\quad a\neq b,
\]
extends to a continuous function on \(\Zp\times\Zp\). The \(C^1\)-norm is
\begin{equation}\label{gC1}
\|g\|_{C^1}=\max\bigl(\|g\|_\infty,\|\Phi_1 g\|_\infty\bigr).
\end{equation}
\end{definition}

\begin{definition}[{\(C^2\)-functions, see \cite[Definition~28.1]{SC}}]
\label{def:C2}
A function \(g:\Zp\to\Cp\) belongs to \(C^2(\Zp\to\Cp)\) if its second difference quotient \(\Phi_2 g\) extends to a continuous function on \(\Zp^3\). The \(C^2\)-norm is
\begin{equation}\label{gC2}
\|g\|_{C^2}=\max\bigl(\|g\|_\infty,\|\Phi_1 g\|_\infty,\|\Phi_2 g\|_\infty\bigr).
\end{equation}
\end{definition}

Recall also that a locally analytic function on \(\Zp\) belongs to
$$C^\infty(\Zp\to\Cp)=\bigcap_{n\ge1}C^n(\Zp\to\Cp);$$ see \cite[Corollary~29.11]{SC}.

\begin{definition}[{Indefinite sum, see \cite[Definition~34.1]{SC}}]
\label{def:indefinite-sum}
For \(g\in C(\Zp\to\Cp)\), the \emph{indefinite sum} of \(g\) is the unique continuous function \(Sg\) satisfying
\[
Sg(x+1)-Sg(x)=g(x),\quad Sg(0)=0.
\]
\end{definition}

By \cite[Theorem~53.6]{SC}, if \(g\in C^1(\Zp\to\Cp)\) then \(Sg\in C^1(\Zp\to\Cp)\) and
\[
\|Sg\|_{C^1}\le p\,\|g\|_{C^1}.
\]
Moreover, if \(g\) is locally analytic then so is \(Sg\), hence \(Sg\in C^2(\Zp\to\Cp)\).

With these definitions at hand, we can state the estimate.

\begin{lemma}[Volkenborn estimate]\label{lem:volkenborn}
Let \(g\in C^2(\Zp\to\Cp)\). Then
\begin{equation}\label{eq:volkenborn-estimate-2}
\left|
\int_{\Zp}g(a)\,\d a-\frac{1}{p^N}\sum_{a=0}^{p^N-1}g(a)
\right|_p
\le
\|Sg\|_{C^2}\,p^{-N},
\end{equation}
where \(Sg\) is the indefinite sum of \(g\). In particular, if
\(\{g_z\}\) is a family of \(C^2\)-functions whose indefinite sums satisfy
\[
\sup_z \|Sg_z\|_{C^2}<\infty,
\]
then the estimate holds with a constant independent of \(z\).
\end{lemma}

\begin{proof}
Since \(g\in C^2(\Zp\to\Cp)\), its indefinite sum \(Sg\) also
belongs to \(C^2(\Zp\to\Cp)\) by
\cite[Corollary~54.3]{SC}. Taylor's formula at \(0\) therefore gives
\begin{equation}\label{eq:taylor-Sg}
Sg(x)=Sg(0)+x(Sg)'(0)+x^2R(x),
\qquad x\in\Zp,
\end{equation}
where
\[
R(x)=\bar\Phi_2(Sg)(x,0,0);
\]
see \cite[Proposition~28.3]{SC}. Here
\(\bar\Phi_2(Sg)\) is the continuous extension of the second
difference quotient \(\Phi_2(Sg)\) to \(\Zp^3\), whose existence
follows from \(Sg\in C^2(\Zp\to\Cp)\)
(Definition~\ref{def:C2}).

The continuity of this extension and the definition of the
\(C^2\)-norm imply that
\begin{equation}\label{Rx}
|R(x)|_p
=
|\bar\Phi_2(Sg)(x,0,0)|_p
\le
\|\Phi_2(Sg)\|_\infty
\le
\|Sg\|_{C^2},
\qquad x\in\Zp.
\end{equation}
Thus the remainder in \eqref{eq:taylor-Sg} is uniformly bounded
on \(\Zp\).

We next express both the Volkenborn integral and its Riemann sums
in terms of \(Sg\):
\begin{equation}\label{indefine}
\int_{\Zp}g(a)\,\d a=(Sg)'(0),
\qquad
\frac{1}{p^N}\sum_{a=0}^{p^N-1}g(a)
=
\frac{Sg(p^N)}{p^N}.
\end{equation}
The first identity is the indefinite-sum representation of the
Volkenborn integral; see \cite[Definition~55.1 and \S54]{SC}.
The second follows by summing
\(Sg(a+1)-Sg(a)=g(a)\) over \(a=0,\ldots,p^N-1\) and using
\(Sg(0)=0\).

Substituting \(x=p^N\) into \eqref{eq:taylor-Sg}, we obtain
\[
Sg(p^N)
=
p^N(Sg)'(0)+p^{2N}R(p^N).
\]
Dividing by \(p^N\) and rearranging yields
\[
(Sg)'(0)-\frac{Sg(p^N)}{p^N}
=
-p^N R(p^N).
\]
By \eqref{indefine}, the left-hand side is precisely the
difference between the Volkenborn integral and its \(N\)th
Riemann sum. Since \(|p^N|_p=p^{-N}\), the bound
\eqref{Rx} gives
\[
\begin{aligned}
\left|
\int_{\Zp}g(a)\,\d a
-
\frac{1}{p^N}\sum_{a=0}^{p^N-1}g(a)
\right|_p
&=
p^{-N}|R(p^N)|_p\\
&\le
\|Sg\|_{C^2}\,p^{-N}.
\end{aligned}
\]
This proves \eqref{eq:volkenborn-estimate-2}.

Finally, suppose that \(\{g_z\}\) is a family of
\(C^2\)-functions satisfying
\[
M:=\sup_z\|Sg_z\|_{C^2}<\infty.
\]
Applying the preceding estimate to each \(g_z\), we find that
\[
\left|
\int_{\Zp}g_z(a)\,\d a
-
\frac{1}{p^N}\sum_{a=0}^{p^N-1}g_z(a)
\right|_p
\le
\|Sg_z\|_{C^2}\,p^{-N}
\le
M\,p^{-N}.
\]
Since \(M\) is independent of \(z\), the estimate is uniform
over the family.
\end{proof}

For the proof, we also need the following lemma.

\begin{lemma}[Decay at infinity for Volkenborn resolvents]
\label{lem:volkenborn-decay}
Let \(\sigma\subset\Cp\) be compact, and let \(f:\Zp\to\sigma\) be
locally analytic thus bounded. Define
\[
F_f(z)=\int_{\Zp}\frac{\d a}{z-f(a)},\quad z\in\Omega_\sigma:=\Cp\setminus\sigma.
\]
Then \(F_f\) is well-defined on \(\Omega_\sigma\), and there exists a
constant \(C>0\), depending only on \(f\), such that
\[
|F_f(z)|_p\le \frac{C}{|z|_p}
\quad(|z|_p\to\infty).
\]
Consequently,
\[
\lim_{|z|_p\to\infty}F_f(z)=0.
\]
\end{lemma}

\begin{proof}
Since \(f\) is bounded, there exists \(M<\infty\) such that
\[
|f(a)|_p\le M\quad(a\in\Zp).
\]
For \(|z|_p>M\) and every \(a\in\Zp\), we have
\[
|z-f(a)|_p=|z|_p.
\]
Define
\[
g_z(a):=\frac{1}{z-f(a)},\quad a\in\Zp.
\]
Then \(g_z\) is locally analytic on \(\Zp\), hence belongs to
\(C^1(\Zp\to\Cp)\). We estimate its \(C^1\)-norm.

First, the supremum norm is
\[
\|g_z\|_\infty
=
\sup_{a\in\Zp}|g_z(a)|_p
=
\sup_{a\in\Zp}\frac{1}{|z-f(a)|_p}
=
\frac1{|z|_p}.
\]

Next, for \(a\ne b\), the first difference quotient is
\[
\Phi_1 g_z(a,b)
=
\frac{g_z(a)-g_z(b)}{a-b}
=
\frac{
\frac{1}{z-f(a)}-\frac{1}{z-f(b)}
}{a-b}.
\]
A direct computation gives
\[
\Phi_1 g_z(a,b)
=
\frac{f(a)-f(b)}{(z-f(a))(z-f(b))(a-b)}
=
\frac{\Phi_1 f(a,b)}{(z-f(a))(z-f(b))}.
\]
Since \(f\) is locally analytic, \(f\in C^1(\Zp\to\Cp)\) and
\[
M_1:=\|\Phi_1 f\|_\infty<\infty.
\]
For \(|z|_p>M\), we have \(|z-f(a)|_p=|z|_p\) and
\(|z-f(b)|_p=|z|_p\). Hence
\[
|\Phi_1 g_z(a,b)|_p
=
\frac{|\Phi_1 f(a,b)|_p}{|z-f(a)|_p\,|z-f(b)|_p}
\le
\frac{M_1}{|z|_p^2}.
\]
Taking the supremum over \(a\ne b\), we obtain
\[
\|\Phi_1 g_z\|_\infty
\le
\frac{M_1}{|z|_p^2}.
\]

Combining the two estimates yields
\[
\|g_z\|_{C^1}
=
\max\bigl(\|g_z\|_\infty,\|\Phi_1 g_z\|_\infty\bigr)
\le
\max\left(
\frac1{|z|_p},
\frac{M_1}{|z|_p^2}
\right).
\]
For \(|z|_p\ge 1\), we have \(1/|z|_p\le 1\) and \(M_1/|z|_p^2\le M_1/|z|_p\). Hence
\[
\|g_z\|_{C^1}
\le
\max\left(\frac{1}{|z|_p},\frac{M_1}{|z|_p^2}\right)
\le
\frac{\max(1,M_1)}{|z|_p}
=
\frac{C_0}{|z|_p},
\]
where \(C_0:=\max(1,M_1)\).

By the Volkenborn estimate for \(C^1\)-functions
(see \cite[Proposition~55.2]{SC}),
\[
\left|\int_{\Zp} g_z(a)\,\d a\right|_p
\le
p\,\|g_z\|_{C^1}.
\]
Therefore
\[
|F_f(z)|_p
=
\left|\int_{\Zp}\frac{\d a}{z-f(a)}\right|_p
\le
p\,\|g_z\|_{C^1}
\le
\frac{pC_0}{|z|_p}.
\]
Setting \(C:=pC_0\), we obtain
\[
|F_f(z)|_p\le \frac{C}{|z|_p}.
\]
Letting \(|z|_p\to\infty\), we conclude \(F_f(z)\to0\).
\end{proof}

\begin{theorem}[{Compatibility; cf. \cite[Definition I.1]{HuKim}}]
\label{thm:compat}
Let \(f:\Zp\to\Cp\) be locally analytic and bounded, and put
\(\sigma=f(\Zp)\). Define
\[
F_f(z)=\int_{\Zp}\frac{\d a}{z-f(a)},\quad z\in\Omega_\sigma.
\]
Then \(F_f\in H_0(\Omega_\sigma)\). Let
\(\mu_{\mathcal T}\in\cL^*(\sigma)\) be the unique generalized
distribution satisfying \(S_{\mu_{\mathcal T}}=F_f\). Then
\[
\zeta_p^{\mu_{\mathcal T}}(s,\lambda)=\zeta_p^f(s,\lambda)
\]
for all \(s\in\Zp\setminus\{1\}\) and all \(\lambda\in\Cp\) with
\(-\lambda\notin\sigma\). In other words, the zeta function of \cite{HuKim} is recovered as the zeta function associated with the spectral distribution \(\mu_{\mathcal T}\in\cL^*(\sigma)\). The fermionic case is analogous, with
\(\mu_{-1}\) in place of the Haar distribution.
\end{theorem}
\begin{proof}
We treat the bosonic case first, then the fermionic case.

\medskip
\noindent\textbf{Bosonic case.}
Let \(f:\Zp\to\Cp\) be locally analytic and bounded. Since \(\Zp\) is
compact and \(f\) is continuous, \(\sigma=f(\Zp)\) is compact.

\medskip
\noindent\textit{Step 1: \(F_f\in H_0(\Omega_\sigma)\).}

Fix \(r>0\) and set
\[
\Omega_\sigma(r)=\{z\in\Cp:\operatorname{dist}(z,\sigma)\ge r\}.
\]
For \(z\in\Omega_\sigma(r)\) and \(a\in\Zp\), we have \(f(a)\in\sigma\),
hence
\begin{equation}\label{fa}
|z-f(a)|_p\ge r.
\end{equation}
Define
\[
g_z(a)=\frac{1}{z-f(a)},\quad a\in\Zp.
\]
Since \(f\) is locally analytic on \(\Zp\), it belongs to
\(C^\infty(\Zp\to\Cp)\) by \cite[Corollary~29.11]{SC}; in particular
\(g_z\in C^2(\Zp\to\Cp)\), and its \(C^2\)-norm admits a bound
independent of \(z\in\Omega_\sigma(r)\). Indeed, for every \(z\in\Omega_\sigma(r)\) and \(a\in\mathbb Z_p\), we have
\[
|g_z(a)|=\frac{1}{|z-f(a)|}\le\frac1r,
\]
so
\[
\|g_z\|_\infty:=\sup_{a\in\mathbb Z_p}|g_z(a)|\le\frac1r.
\]
For \(a\neq b\),
\[
\Phi_1 g_z(a,b)
=\frac{g_z(a)-g_z(b)}{a-b}
=\frac{\Phi_1 f(a,b)}{(z-f(a))(z-f(b))},
\]
hence
\[
\|\Phi_1 g_z\|_\infty
:=\sup_{\substack{a,b\in\mathbb Z_p\\a\neq b}}|\Phi_1 g_z(a,b)|
\le\frac{\|f\|_{C^1}}{r^2}.
\]
By a similar argument applied to the second difference quotient, we have
\[
\|\Phi_2 g_z\|_\infty
:=\sup_{\substack{a,b,c\in\mathbb Z_p\\a,b,c\text{ pairwise distinct}}}
|\Phi_2 g_z(a,b,c)|
\le\frac{C_f}{r^3},
\]
where \(C_f\) depends only on \(\|f\|_{C^2}\). Consequently, by (\ref{gC2})
\begin{equation}\label{gz}
\sup_{z\in\Omega_\sigma(r)}\|g_z\|_{C^2}
=
\sup_{z\in\Omega_\sigma(r)}
\max\Bigl(
\|g_z\|_\infty,
\|\Phi_1 g_z\|_\infty,
\|\Phi_2 g_z\|_\infty
\Bigr)
\le C(r)<\infty,
\end{equation}
where the inner suprema in the definitions of \(\|g_z\|_\infty\), \(\|\Phi_1 g_z\|_\infty\), and \(\|\Phi_2 g_z\|_\infty\) are taken over \(a,b,c\in\mathbb Z_p\) (with the indicated distinctness conditions), while the outer supremum is taken over the parameter \(z\in\Omega_\sigma(r)\).

By \cite[Corollary~54.3]{SC}, the indefinite sum maps
\(C^2(\Zp\to\Cp)\) into itself, and the operator
\(S:C^2(\Zp\to\Cp)\to C^2(\Zp\to\Cp)\) is continuous (by the closed
graph theorem; see \cite[Theorem 3.5]{vanRooij}). Combining this with
\eqref{gz}, we obtain
\[
\sup_{z\in\Omega_\sigma(r)}\|Sg_z\|_{C^2}
\le
\|S\|_{C^2\to C^2}\,C(r)
=:
C_r<\infty.
\]
Lemma~\ref{lem:volkenborn} therefore applies to each \(g_z\) with the
same bound, giving
\[
|F_f(z)-F_{f,N}(z)|_p
\le
\|Sg_z\|_{C^2}\,p^{-N}
\le
C_r\,p^{-N},
\]
where \(F_{f,N}(z):=p^{-N}\sum_{a=0}^{p^N-1}(z-f(a))^{-1}\). Since the
right-hand side is independent of \(z\), \(F_{f,N}\to F_f\) uniformly
on \(\Omega_\sigma(r)\). Each \(F_{f,N}\) is a rational function with
poles in \(\sigma\), so \(F_f\) is Krasner analytic on \(\Omega_\sigma\).

It remains to verify that \(F_f\) vanishes at infinity. This is
precisely the content of Lemma~\ref{lem:volkenborn-decay}: since \(f\)
is locally analytic and bounded, the lemma yields a constant \(C>0\),
depending only on \(f\), such that
\[
|F_f(z)|_p\le \frac{C}{|z|_p}
\]
for all sufficiently large \(|z|_p\). Consequently,
\[
F_f(z)\longrightarrow 0
\quad(|z|_p\to\infty).
\]
Therefore \(F_f\in H_0(\Omega_\sigma)\).

Thus by Vishik's inverse Stieltjes theorem (Theorem~\ref{thm:invStieltjes}),
there exists a unique generalized distribution
\[
\mu_{\mathcal T}\in\cL^*(\sigma)
\]
such that
\[
S_{\mu_{\mathcal T}}=F_f.
\]

\medskip
\noindent\textit{Step 2: The evaluation of \(\mu_{\mathcal T}(g_\lambda)\).}

Fix \(s\in\Zp\setminus\{1\}\) and \(\lambda\in\Cp\) with
\(-\lambda\notin\sigma\), and set
\[
g_\lambda(x):=\langle\lambda+x\rangle^{1-s},\quad x\in\sigma.
\]
Since \(-\lambda\notin\sigma\) and \(\sigma\) is compact, we have
\(\delta_\lambda:=\operatorname{dist}(-\lambda,\sigma)>0\), so
\(\lambda+y\neq0\) for every \(y\in\mathbb D_\sigma(\delta_\lambda)\).
By Subsection~\ref{sec:p-adic-numbers}, the map
\(y\mapsto\langle\lambda+y\rangle^{1-s}\) is locally analytic wherever
\(\lambda+y\neq0\). Arguing exactly as in the proof of
Lemma~\ref{lem:g-lambda-s}(i), which uses only the fact that
\(\lambda+x\neq0\) on \(\sigma\), we find \(r>0\) such that
\(g_\lambda\in B_r(\sigma)\). Hence \(g_\lambda\in\cL(\sigma)\).

Let \(\mu_{\mathcal T}\) be the generalized distribution from Step 1. By the explicit form of the inverse Stieltjes transform (Theorem~\ref{thm:invStieltjes}, (\ref{VF})), substituting \(F_f(x)\) for \(F(x)\),
\begin{equation}\label{muT}
\mu_{\mathcal T}(g_\lambda)
=
\sum_i\int_{a_i,\Gamma_i}
F_f(x)\,g_\lambda(x)\,(x-a_i)\,\d x,
\end{equation}
where the \(a_i,\Gamma_i\) are as in Theorem~\ref{thm:invStieltjes}.

We shall interchange the Shnirelman integral with the Volkenborn integral
defining \(F_f\) in (\ref{muT}). To do this, for each \(i\), put \(r_i:=|\Gamma_i|_p\), and let
\(x\) lie on the circle \(|x-a_i|_p=r_i\). We claim that this circle
is contained in \(\Omega_\sigma(r_i)\). Indeed, for every \(a\in\sigma\),
the element \(a\) lies in one of the discs
\(D_j=\{y:|y-a_j|_p<r_i\}\) covering \(\sigma\). If \(j=i\), then
\(|a-a_i|_p<r_i\), so \(|x-a|_p=|x-a_i|_p=r_i\). If \(j\ne i\), the
discs \(D_i\) and \(D_j\) are disjoint, so \(|a_i-a_j|_p\ge r_i\); if
we had \(|x-a|_p<r_i\), then
\[
|a_i-a_j|_p\le\max(|a_i-a|_p,|a-a_j|_p)<r_i,
\]
a contradiction. Hence \(|x-a|_p\ge r_i\), i.e.,
\(x\in\Omega_\sigma(r_i)\).

For each fixed \(i\), define
\[
h_i(a):=\int_{a_i,\Gamma_i}
\frac{g_\lambda(x)\,(x-a_i)}{x-f(a)}\,\d x,
\quad a\in\mathbb Z_p.
\]
Since the circle \(|x-a_i|_p=r_i=|\Gamma_i|_p\) is contained in
\(\Omega_\sigma(r_i)\), by (\ref{fa}) we have
\[
|x-f(a)|_p\ge r_i>0
\]
for all \(x\) on the circle and all \(a\in\mathbb Z_p\). Hence the map
\[
a\longmapsto \frac{1}{x-f(a)}
\]
is locally analytic on \(\mathbb Z_p\). Consequently, \(h_i\) is locally
analytic on \(\mathbb Z_p\). In particular, by \cite[Corollary~29.11]{SC}, locally analytic functions are \(C^\infty\), hence
\begin{equation}\label{hi}
h_i\in C^\infty(\mathbb Z_p\to\mathbb C_p)\subset C^1(\mathbb Z_p\to\mathbb C_p),
\end{equation}
so its Volkenborn integral exists by \cite[Definition~55.1]{SC}.

By the Volkenborn estimate of Step~1, the partial sums
\[
F_{f,N}(x)=\frac{1}{p^N}\sum_{a=0}^{p^N-1}\frac{1}{x-f(a)}
\]
satisfy
\[
\sup_{x\in\Omega_\sigma(r_i)}|F_f(x)-F_{f,N}(x)|_p\le C_{r_i}\,p^{-N}.
\]
Since the circle \(|x-a_i|_p=r_i\) is contained in \(\Omega_\sigma(r_i)\),
the convergence is uniform on the circle.
By Proposition~\ref{prop:shnirelman-uniform}(1), the Shnirelman integral
is compatible with this uniform limit:
\[
\int_{a_i,\Gamma_i}F_{f,N}(x)\,g_\lambda(x)\,(x-a_i)\,\d x
\longrightarrow
\int_{a_i,\Gamma_i}F_f(x)\,g_\lambda(x)\,(x-a_i)\,\d x.
\]

On the other hand, because
\[
F_{f,N}(x)=\frac{1}{p^N}\sum_{a=0}^{p^N-1}\frac{1}{x-f(a)},
\]
we have
\[
\int_{a_i,\Gamma_i}F_{f,N}(x)\,g_\lambda(x)\,(x-a_i)\,\d x
=
\frac{1}{p^N}\sum_{a=0}^{p^N-1}
\int_{a_i,\Gamma_i}
\frac{g_\lambda(x)\,(x-a_i)}{x-f(a)}\,\d x
=
\frac{1}{p^N}\sum_{a=0}^{p^N-1}h_i(a).
\]
Since $h_i$ are Volkenborn integrable as stated in the above, by the definition of the Volkenborn integral,
\[
\frac{1}{p^N}\sum_{a=0}^{p^N-1}h_i(a)
\longrightarrow
\int_{\mathbb Z_p}h_i(a)\,\d a
\quad(N\to\infty).
\]
Combining these limits, we obtain
\[
\int_{a_i,\Gamma_i}F_f(x)\,g_\lambda(x)\,(x-a_i)\,\d x
=
\int_{\mathbb Z_p}
\left(
\int_{a_i,\Gamma_i}
\frac{g_\lambda(x)\,(x-a_i)}{x-f(a)}\,\d x
\right)\d a.
\]
Summing over \(i\), (\ref{muT}) implies
\[
\mu_{\mathcal T}(g_\lambda)
=
\int_{\Zp}
\left(
\sum_i\int_{a_i,\Gamma_i}
\frac{g_\lambda(x)\,(x-a_i)}{x-f(a)}\,\d x
\right)\d a.
\]

Then for each fixed \(a\in\Zp\), the element \(f(a)\) lies in \(\sigma\), so
the \(p\)-adic Cauchy integral formula (Proposition~\ref{prop:cauchy}) applied to \(g_\lambda\in\cL(\sigma)\)
gives
\[
\sum_i\int_{a_i,\Gamma_i}
\frac{g_\lambda(x)\,(x-a_i)}{x-f(a)}\,\d x
=
g_\lambda(f(a)).
\]
Therefore
\[
\mu_{\mathcal T}(g_\lambda)
=
\int_{\Zp}g_\lambda(f(a))\,\d a.
\]

\medskip
\noindent\textit{Step 3: Conclusion.}

By Definition~\ref{def:bosonic-zeta} and the definition of
\(\zeta_p^f\),
\[
\begin{aligned}
\zeta_p^{\mu_{\mathcal T}}(s,\lambda)
&= \frac{1}{s-1}\,\mu_{\mathcal T}(g_\lambda)
= \frac{1}{s-1}\int_{\Zp}g_\lambda(f(a))\,\d a \\
&= \frac{1}{s-1}\int_{\Zp}\langle\lambda+f(a)\rangle^{1-s}\,\d a
= \zeta_p^f(s,\lambda).
\end{aligned}
\]

\medskip
\noindent\textbf{Fermionic case.}
Let \(f:\Zp\to\sigma\) be continuous. Since \(\mu_{-1}\) is bounded,
\[
|\mu_{-1}(a+p^N\Zp)|_p=|(-1)^a|_p=1
\]
for all \(a,N\) (Definition~\ref{def:mu-1}), the measure \(\mu_{-1}\)
already belongs to \(\cL^*(\sigma)\) (Definition~\ref{def:generalized-distribution}).
Define
\[
\nu_f^E(g)=\int_{\Zp}g(f(a))\,\d\mu_{-1}(a),\quad g\in\cL(\sigma).
\]
For all \(g\in B_r(\sigma)\),
\[
|\nu_f^E(g)|_p\le\sup_{a\in\Zp}|g(f(a))|_p\le\|g\|_r,
\]
so \(\nu_f^E\in\cL^*(\sigma)\) with \(\|\nu_f^E\|_r\le1\).

Its Stieltjes transform is
\[
S_{\nu_f^E}(z)
=
\int_{\Zp}\frac{\d\mu_{-1}(a)}{z-f(a)}
=
F_f^E(z),
\]
which is the fermionic resolvent of
Definition~\ref{def:fermionic-resolvent}. Since \(\nu_f^E\in\cL^*(\sigma)\) and \(S_{\nu_f^E}=F_f^E\), the inverse
Stieltjes theorem implies \(F_f^E\in H_0(\Omega_\sigma)\) and that
\(F_f^E\) is the Stieltjes transform of \(\nu_f^E\). Hence, using the
definition of the fermionic spectral zeta function
(Definition~\ref{def:fermionic-zeta}),
\[
\zeta_{p,E}^{\nu_f^E}(s,\lambda)
=
\int_\sigma\langle\lambda+x\rangle^{1-s}\,\d\nu_f^E(x)
=
\int_{\Zp}\langle\lambda+f(a)\rangle^{1-s}\,\d\mu_{-1}(a)
=
\zeta_{p,E}^f(s,\lambda).
\]
This proves the fermionic compatibility.
\end{proof}

\section{$p$-adic functional determinants}\label{sec:det}

The spectral zeta functions introduced in Section~\ref{sec:zeta} give rise, by differentiation at $s=0$, to the $p$-adic functional determinants. In this section we shall develop their theory for a general spectral distribution $\mu\in\cL^*(\sigma)$. In the classical archimedean setting, the functional determinant is obtained by differentiating the spectral zeta function at $s=0$, and it encodes the one-loop effective action of a quantum system \cite{Hawking1977,Voros1992}; see also \cite{Elizalde,Kirsten} for comprehensive accounts of its physical applications, including quantum field theory, spectral geometry and Casimir energies. Recent developments in the archimedean setting include the localization of spectral zeta functions and determinants on Lie groups \cite{Spreafico}, systematic regularization techniques for functional determinants in quantum field theory \cite{ShojiYamaguchi2025}, and the application of supersymmetric zeta functions and determinants to supersymmetric indices and Casimir energies \cite{NakayamaOkazaki2026}.

Our previous framework established the corresponding $p$-adic construction for discrete spectra admitting a locally analytic interpolation function. Our goal here is to extend that construction to the more general spectral distributions introduced in Section~\ref{sec:zeta}. We first define the bosonic and fermionic functional determinants, then prove their integral representations, and finally derive Stirling-type expansions in powers of $1/\lambda$. Throughout this section, we assume that $\mu$ and $\nu$ are generalized distributions obtained from suitable resolvents via the inverse Stieltjes transform.

\begin{definition}\label{def:functional-det}
For a bosonic spectral distribution $\mu$, the \emph{$p$-adic functional determinant} is
\[
\log\Gamma_p^\mu(\lambda)
=
\omega_v(\lambda)
\frac{\partial}{\partial s}
\zeta_p^\mu(s,\lambda)
\bigg|_{s=0}.
\]
For a fermionic distribution $\nu$, define analogously
\[
\log\Gamma_{p,E}^\nu(\lambda)
=
\omega_v(\lambda)
\frac{\partial}{\partial s}
\zeta_{p,E}^\nu(s,\lambda)
\bigg|_{s=0}.
\]
\end{definition}
\begin{remark}
This is a generalization of \cite[Definition~IV.1]{HuKim}, which was formulated for discrete spectra. In the discrete case with a locally analytic interpolation function, the definition of \cite{HuKim} is recovered as a special case.
\end{remark}
\begin{theorem}[Integral representation]\label{thm:integral-rep}
For $\lambda$ sufficiently large and $\lambda+x$ avoiding zero,
\[
\log\Gamma_p^\mu(\lambda)
=
\int_\sigma
(\lambda+x)\bigl(\log_p(\lambda+x)-1\bigr)
\,\d\mu(x).
\]
For the fermionic case,
\[
\log\Gamma_{p,E}^\nu(\lambda)
=
-\int_\sigma
(\lambda+x)\log_p(\lambda+x)
\,\d\nu(x).
\]
\end{theorem}

\begin{proof}
Let us first consider the bosonic case. By Definition~\ref{def:bosonic-zeta},
\[
(s-1)\zeta_p^\mu(s,\lambda)
=
\int_\sigma \langle\lambda+x\rangle^{1-s}\,\d\mu(x).
\]
Differentiating both sides with respect to \(s\) and setting \(s=0\), we obtain
\[
\zeta_p^\mu(0,\lambda)
-
\frac{\partial}{\partial s}\zeta_p^\mu(s,\lambda)\bigg|_{s=0}
=
-\int_\sigma \langle\lambda+x\rangle
\log_p\langle\lambda+x\rangle
\,\d\mu(x),
\]
where we used \(\frac{\partial}{\partial s}\langle\lambda+x\rangle^{1-s}=-\langle\lambda+x\rangle^{1-s}\log_p\langle\lambda+x\rangle\). Equivalently,
\[
\frac{\partial}{\partial s}\zeta_p^\mu(s,\lambda)\bigg|_{s=0}
=
\int_\sigma \langle\lambda+x\rangle
\bigl(\log_p\langle\lambda+x\rangle-1\bigr)
\,\d\mu(x),
\]
since \(\zeta_p^\mu(0,\lambda)=-\int_\sigma \langle\lambda+x\rangle\,d\mu(x)\) by the same definition with \(s=0\).

Now, for \(|\lambda|_p\) sufficiently large and \(-\lambda\notin\supp(\mu)\), we have \(|\lambda+x|_p=|\lambda|_p\) for all \(x\in\supp(\mu)\), so that
\[
\lambda+x=\omega_v(\lambda+x)\langle\lambda+x\rangle,
\quad
\omega_v(\lambda+x)=\omega_v(\lambda).
\]
Moreover, the Iwasawa logarithm \(\log_p\) vanishes on roots of unity, hence
\[
\log_p(\lambda+x)=\log_p\langle\lambda+x\rangle.
\]
Combining these identities with the definition
\[
\log\Gamma_p^\mu(\lambda)
=
\omega_v(\lambda)
\frac{\partial}{\partial s}\zeta_p^\mu(s,\lambda)\bigg|_{s=0},
\]
we obtain
\[
\log\Gamma_p^\mu(\lambda)
=
\int_\sigma
(\lambda+x)\bigl(\log_p(\lambda+x)-1\bigr)
\,\d\mu(x),
\]
which is the desired formula.

The fermionic case is analogous. By Definition~\ref{def:fermionic-zeta},
\[
\zeta_{p,E}^\nu(s,\lambda)
=
\int_\sigma \langle\lambda+x\rangle^{1-s}\,\d\nu(x).
\]
Differentiating with respect to \(s\) and setting \(s=0\),
\[
\frac{\partial}{\partial s}\zeta_{p,E}^\nu(s,\lambda)\bigg|_{s=0}
=
-\int_\sigma \langle\lambda+x\rangle
\log_p\langle\lambda+x\rangle
\,\d\nu(x).
\]
Multiplying by \(\omega_v(\lambda)\) and using the same identities as above, we arrive at
\[
\log\Gamma_{p,E}^\nu(\lambda)
=
-\int_\sigma
(\lambda+x)\log_p(\lambda+x)
\,\d\nu(x).
\]
This completes the proof.
\end{proof}

Having established the integral representation of the functional determinant in Theorem~\ref{thm:integral-rep}, we now derive its asymptotic expansion for large \(|\lambda|_p\). This expansion is the \(p\)-adic analogue of Stirling's series for the logarithm of the gamma function.

\begin{theorem}[Stirling's series]\label{thm:stirling}
Let $\mu\in\cL^*(\sigma)$ and let $|\lambda|_p>\max_{x\in\supp(\mu)}|x|_p$. Then
\[
\log\Gamma_p^\mu(\lambda)
=
B_1^\mu(0)
+
\sum_{n=1}^{\infty}
\frac{(-1)^{n+1}}{n(n+1)}
B_{n+1}^\mu(0)
\frac{1}{\lambda^n}
+
B_1^\mu(\lambda)\log_p\lambda
-
B_1^\mu(\lambda).
\]
\end{theorem}
\begin{proof}
For large $|\lambda|_p$, we have $|x/\lambda|_p<1$ for all $x\in\supp(\mu)$, so the power series expansion of $\log_p(1+T)$ gives
\[
\log_p(\lambda+x)
=
\log_p\lambda+\log_p\!\left(1+\frac{x}{\lambda}\right)
=
\log_p\lambda+\sum_{n=1}^{\infty}\frac{(-1)^{n+1}}{n}\frac{x^n}{\lambda^n}.
\]
Integrating term by term against $\mu$ and rearranging yields
\[
\log\Gamma_p^\mu(\lambda)
=
B_1^\mu(0)
+
\sum_{n=1}^{\infty}
\frac{(-1)^{n+1}}{n(n+1)}
B_{n+1}^\mu(0)
\frac{1}{\lambda^n}
+
B_1^\mu(\lambda)\log_p\lambda
-
B_1^\mu(\lambda),
\]
where we used the definition $B_m^\mu(\lambda)=\int_\sigma(\lambda+x)^m\,d\mu(x)$ and the identity $x\log_p\lambda=\lambda\log_p\lambda\cdot(x/\lambda)$. The details are otherwise the same as in our previous framework (see \cite[Theorem IV.3]{HuKim}).
\end{proof}
\begin{remark}
The moments $B_m^\mu(0)$ are exactly the coefficients in the Laurent expansion of the resolvent $F(z)$ at infinity:
\[
F(z)=\sum_{m=0}^{\infty}\frac{B_m^\mu(0)}{z^{m+1}}.
\]
Thus the functional determinant is completely determined by the asymptotic behavior of the resolvent.
\end{remark}

\section{Application: The position operator in $p$-adic quantum mechanics}
\label{sec:applications}

In this section we illustrate the general framework of
Sections~\ref{sec:zeta} and~\ref{sec:det} with a single example that
satisfies all hypotheses of the inverse Stieltjes theorem
(Theorem~\ref{thm:invStieltjes}) unconditionally. The example is the
position operator in \(p\)-adic quantum mechanics, restricted to an
arbitrary compact subset \(\sigma\subset\Cp\). This operator has a
genuinely continuous compact spectrum, and its resolvent is explicit,
which allows a direct computation of the \(p\)-adic spectral
distribution, the zeta functions, and the functional determinant.

The position operator in \(p\)-adic quantum mechanics is defined as the
multiplication operator
\[
(Q\psi)(x)=x\psi(x),
\]
acting on \(p\)-adic valued wave functions; see
\cite[Eq.~(5.3)]{Khrennikov1991} and
\cite[Chapter~IV, \S1]{Khrennikov1994book}. In the \(p\)-adic Hilbert
space \(L_2(\Qp^n,\nu(dx))\) introduced in
\cite{Khrennikov1991,Khrennikov1994book}, it is a symmetric operator
with continuous spectrum. For our purposes, it is convenient to work
instead with the Banach space \(C(\sigma\to\Cp)\) on a compact subset
\(\sigma\subset\Cp\), on which the resolvent is bounded and the inverse
Stieltjes theorem applies directly. The two settings are compatible at
the level of the scalar resolvent: the Stieltjes transform
\(F_Q(z)=\int_\sigma d\mu(x)/(z-x)\) and the spectral distribution
\(\mu\) obtained below depend only on the measure \(\mu\) and the
compact set \(\sigma\), not on the choice of the underlying Banach
space.

Its spectrum on \(\Qp\) is all of \(\Qp\), which is not compact in the
\(p\)-adic topology. To obtain a compact spectral set, we choose an
arbitrary compact subset \(\sigma\subset\Cp\) and consider the position
operator restricted to \(\sigma\). A natural example is the closed unit
ball of a finite extension \(K/\Qp\), which can be viewed as a \(p\)-adic
analogue of a bounded spatial region. We may take, for instance,
\(\sigma=\{x\in K:|x|_p\le1\}\) for a finite extension \(K/\Qp\).

Let
\[
X=C(\sigma\to\Cp)
\]
be the \(\Cp\)-Banach space of continuous functions on \(\sigma\),
equipped with the supremum norm
\(\|f\|_\infty=\sup_{x\in\sigma}|f(x)|_p\). Define the position
operator
\[
(Qf)(x)=x f(x),\quad f\in X,\ x\in\sigma.
\]
Then \(Q\) is bounded with
\(\|Q\|_p\le\sup_{x\in\sigma}|x|_p<\infty\), and its spectrum is
\[
\sigma(Q)=\sigma,
\]
which is compact and continuous (i.e.\ uncountable, with no isolated
points). For \(z\notin\sigma\), the resolvent is
\[
(R(z)f)(x)=\frac{f(x)}{z-x},\quad f\in X,\ x\in\sigma.
\]

Choose a bounded \(p\)-adic measure \(\mu\) on \(\sigma\), and define
the continuous linear functional
\[
\ell_\mu(f)=\int_\sigma f\,\d\mu,\quad f\in X.
\]
This is the functional \(\ell\) required in
Definition~\ref{def:admissible-operator}. Then the scalar resolvent is
\[
F_Q(z)=\ell_\mu(R(z)\mathbf{1})
=
\int_\sigma\frac{\d\mu(x)}{z-x},
\]
where \(\mathbf{1}\) denotes the constant function \(1\). Since \(\mu\)
is bounded and \(|z-x|_p\ge\dist(z,\sigma)>0\) for \(z\in\Omega_\sigma\),
the function \(F_Q\) is Krasner analytic on \(\Omega_\sigma\), and
\(F_Q(z)\to0\) as \(|z|_p\to\infty\). Hence
\(F_Q\in H_0(\Cp\setminus\sigma)\).

By the inverse Stieltjes theorem (Theorem~\ref{thm:invStieltjes}),
there exists a unique generalized distribution
\(\mu_Q\in\cL^*(\sigma)\) such that
\[
F_Q(z)=\int_\sigma\frac{\d\mu_Q(x)}{z-x}.
\]
But \(F_Q\) is by construction the Stieltjes transform of \(\mu\), so
by uniqueness,
\[
\mu_Q=\mu.
\]
Thus the \(p\)-adic spectral distribution of the position operator
\(Q\) is exactly the chosen measure \(\mu\). In particular, if \(\mu\)
is the normalized Haar measure on \(\sigma\), then \(F_Q\) is the
Stieltjes transform of the Haar measure; if \(\mu=\delta_{x_0}\) is a
point mass, then \(F_Q(z)=1/(z-x_0)\).

The bosonic and fermionic \(p\)-adic spectral zeta functions are
\[
\zeta_p^{\mu_Q}(s,\lambda)
=
\frac{1}{s-1}\int_\sigma\langle\lambda+x\rangle^{1-s}\,\d\mu(x),
\quad
\zeta_{p,E}^{\nu_Q}(s,\lambda)
=
\int_\sigma\langle\lambda+x\rangle^{1-s}\,\d\mu(x).
\]
Their special values at non-positive integers are given by
\[
\zeta_p^{\mu_Q}(1-m,\lambda)
=
-\frac{1}{\omega_v^m(\lambda)}
\frac{1}{m}\int_\sigma(\lambda+x)^m\,\d\mu(x),
\]
\[
\zeta_{p,E}^{\nu_Q}(1-m,\lambda)
=
\frac{1}{\omega_v^m(\lambda)}
\int_\sigma(\lambda+x)^m\,\d\mu(x).
\]
The functional determinant \(\log\Gamma_p^{\mu_Q}(\lambda)\) is defined
as in Section~\ref{sec:det}, with integral representation
\[
\log\Gamma_p^{\mu_Q}(\lambda)
=
\int_\sigma(\lambda+x)(\log_p(\lambda+x)-1)\,\d\mu(x),
\]
and Stirling expansion of Theorem~\ref{thm:stirling} with moments
\[
B_m^{\mu_Q}(0)=\int_\sigma x^m\,\d\mu(x).
\]

\begin{remark}
In \(p\)-adic quantum mechanics, \(Q\) is the position operator of a
\(p\)-adic particle confined to a compact spatial region, and the
\(p\)-adic spectral zeta function and functional determinant provide a
regularization of the spectral data of this operator in the \(p\)-adic
setting.
\end{remark}

\begin{remark}\label{rem:why-not-HK}
The example above reveals a structural difference from the framework of
\cite{HuKim}. There, the zeta function is defined by
\[
\zeta_p^f(s,\lambda)
=
\frac{1}{s-1}\int_{\Zp}\langle\lambda+f(a)\rangle^{1-s}\,\d a,
\]
where \(f\) interpolates the integer spectrum
\(\{\lambda_n\}_{n\in\mathbb N_0}\). Since \(\mathbb N_0\) is dense in
\(\Zp\), the integral is taken over the parameter space \(\Zp\) with
the Haar distribution. But a general spectrum is a compact set
\(\sigma\subset\Cp\), not necessarily \(\Zp\); hence the construction
is tied to \(\Zp\) and does not carry over to \(\sigma\) in general.

More intrinsically, the Haar distribution is an a priori measure on
\(\Zp\) and bears no direct relation to the operator's spectral
distribution. In contrast, here the spectral distribution
\(\mu\in\cL^*(\sigma)\) is built directly from the resolvent
\(R(z)=(zI-T)^{-1}\) via the inverse Stieltjes transform. For the
position operator, \(\mu\) is a bounded measure on \(\sigma\), and its
Stieltjes transform is exactly the scalar resolvent
\(F_Q(z)=\int_\sigma d\mu(x)/(z-x)\); the inverse Stieltjes transform
then recovers \(\mu\in\cL^*(\sigma)\). Thus the present framework makes
essential use of the operator's spectral distribution. This is why it
can handle genuinely continuous spectra in a more general way.
\end{remark}

\end{document}